\documentclass[12pt,a4paper]{article}

\usepackage{amssymb,amsmath}
\usepackage{graphicx,graphics}

\usepackage[english]{babel}
\usepackage{epsfig,url}
\usepackage{bbm,theorem}
\usepackage{a4wide}

\newtheorem{theorem}{Theorem}[section]

\newtheorem{corollary}[theorem]{Corollary}
\newtheorem{lemma}[theorem]{Lemma}
\newtheorem{proposition}[theorem]{Proposition}
\newtheorem{assumption}[theorem]{Assumption}

{\theorembodyfont{\upshape}

}
\numberwithin{equation}{section}
\numberwithin{theorem}{section}

\newcommand{\qed}{\hfill$\Box$}
\newenvironment{proof}{\begin{trivlist}\item[]{\em Proof:}\/}{%
\qed\end{trivlist}}
\newenvironment{proofof}[1]{%
\begin{trivlist}\item[]{\em Proof of #1}\ }{\qed\end{trivlist}}

\newcommand{\Z}{{\mathbb Z}}
\newcommand{\R}{{\mathbb R}}
\newcommand{\C}{{\mathbb C\hspace{0.05 ex}}}
\newcommand{\N}{{\mathbb N}}
\newcommand{\T}{{\mathbb T}}
\newcommand{\cf}{{\mathbbm 1}} 

\newcommand{\ci}{{\rm i}}
\newcommand{\re}{{\rm Re\,}}

\newcommand{\rme}{{\rm e}}
\newcommand{\rmd}{{\rm d}}

\newcommand{\FT}[1]{\widehat{#1}}

\DeclareMathOperator*{\sign}{sign}

\newcommand{\norm}[1]{\Vert #1\Vert}
\newcommand{\defset}[2]{ \left\{ #1\left|\,
 #2\makebox[0cm]{$\displaystyle\phantom{#1}$}\right.\!\right\} }
\newcommand{\set}[1]{\{#1\}}

\newcommand{\mean}[1]{\langle #1\rangle}
\newcommand{\vep}{\varepsilon}
\newcommand{\defem}[1]{{\em #1\/}}

\newcommand{\anhH}{H^{{\rm anh}}_L}

\newcounter{jlisti}

\newcommand{\myfigure}[2]{ \includegraphics*[#1]{#2} }

\newcommand{\rphi}{r_\Phi}

\newcommand{\Tobs}{T^{(\text{obs})}}
\newcommand{\Tpred}{T^{(\text{pred})}}

\newcommand{\email}[1]{E-mail: \tt #1}
\newcommand{\emailjani}{\email{jani.lukkarinen@helsinki.fi}}
\newcommand{\addressjani}{\em University of Helsinki,
Department of Mathematics and Statistics\\
\em P.O. Box 68,
FI-00014 Helsingin yliopisto, Finland}

\title{Thermalization in harmonic particle chains with velocity flips}
\author{Jani Lukkarinen\thanks{\emailjani}\\[1em]
$^*$\addressjani}
\date{\today}

\begin{document}

\selectlanguage{english}
\maketitle

\begin{abstract} 
We propose a new mathematical tool for the study of transport properties of
models for lattice vibrations in crystalline solids.  By replication of
dynamical degrees of freedom, we aim at a new dynamical system where the
``local'' dynamics can be isolated and solved independently from the ``global''
evolution.  The replication procedure is very generic but not unique as it
depends on how the original dynamics are split between the local and global
dynamics.  As an explicit example, we apply the scheme to study thermalization
of the pinned harmonic chain with velocity flips.  We improve on the previous
results about this system by showing that after a relatively short time period
the average kinetic temperature profile satisfies the dynamic Fourier's law in a
local
microscopic sense without assuming that the initial data is close to a local
equilibrium state.  The bounds derived here prove that the above thermalization
period is at most of the order $L^{2/3}$, where $L$ denotes the number of
particles in the chain.  
In particular, even before the diffusive time scale Fourier's law becomes a
valid approximation of the evolution of the kinetic temperature profile.
As a second application of the dynamic replica method, we also briefly discuss
replacing the velocity flips by
an anharmonic onsite potential.
\end{abstract}

\vskip 1em 

\begin{center}
\textit{Dedicated to Herbert Spohn, with sincere gratitude for his support,\\
inspiration and insight.}
\end{center}

\newpage

\section{Introduction}
\label{sec:intro}

The energy transport properties of lattice vibrations in
crystalline solids have recently attracted much research activity.  In the
simplest 
case energy is the only relevant conserved quantity, and then it is expected
that for a large class of three or higher dimensional systems energy transport
is diffusive and the Fourier's law of heat conduction holds; see for instance 
\cite{bonetto00} for a discussion.  In contrast, many
one-dimensional systems exhibit anomalous energy transport
violating the Fourier's law.
Section 7 of Ref.~\cite{spohn13} offers a concise summary of the state of the
art in the results and understanding of transport properties of such systems.

If a suitable stochastic noise is added to the Hamiltonian interactions, also
one-dimen\-sional particle chains can produce diffusive
energy
transport.  A particularly appealing test case is obtained by taking a harmonic
chain, which has  ballistic energy transport \cite{lebo67}, and endowing each
of the  
particles with its own Poissonian clock whose rings will flip the velocity of
the particle.  
To our knowledge, this \defem{velocity flip model} was first
considered in \cite{FFL94}, and it is one of the simplest known particle chain
models which has a finite heat conductivity and satisfies the time-dependent
Fourier's law.  Its transport properties depend on the harmonic interactions,
most importantly on whether the forces have an on-site component
(\defem{pinning}) 
or not.  For nearest neighbor interactions, 
if there is no pinning, there are two locally conserved fields, while with
pinning there is only one, the energy density.  In addition, the thermal
conductivity, and hence the energy diffusion constant, happens to be independent
of temperature, which implies that the Fourier's law corresponds to a
\defem{linear} 
heat equation.  This allows for explicit representation of its solutions 
in terms of Fourier transform.  It also leads to the useful
simplification that for any stochastic initial state also the expectation value
of
the temperature distribution satisfies the Fourier's law, even when
the initial total energy has macroscopic variation.

The main goal of the present contribution is to introduce a mathematical method
which allows splitting the dynamics of the velocity flip model into local and
global components in a 
controlled manner.  The local evolution can then be chosen conveniently to
simplify the analysis. For instance, here we show how to apply the method to
separate the harmonic interactions and a dissipation term generated by the noise
into explicitly solvable local dynamics.  Although created with perturbation
theory in mind, the method itself is non-perturbative.
In fact, our main result will assume the exact opposite: we work in the
regime in which the noise dominates over the harmonic evolution, as then it will be
possible to neglect certain resonant terms which otherwise would require more
involved analysis.

The method of splitting is quite generic but in the end it only amounts to
reorganization of 
the original dynamics.  Therefore, it is important to demonstrate that it
can become a useful tool also in practice.  As a test case, we study here the
velocity flip model with pinning and in the regime where the noise dominates,
i.e., the flipping rate is high enough.  The ultimate goal is to prove that this
system \defem{thermalizes} for any sensible initial data: after some initial
time $t_0$ any local correlation function, i.e., an expectation value of
a polynomial of positions and velocities of particles microscopically
close to some given point, is well approximated by the corresponding
expectation taken over some 
statistical equilibrium (thermal) ensemble.  The thermalization time $t_0$ may
depend on the initial data and on the system size but for systems with
``normal'' transport this time should be less than diffusive, $t_0\ll L^2$,
where $L$ denotes the length of the chain. 

We do not have a proof of such a strong statement yet and we only indicate how
the present methods should help in arriving at such conclusions.  Instead of the
full local statistics, we focus here on the time evolution of the
\defem{average kinetic temperature profile}, the observables $\mean{p(t)_x^2}$
where $p(t)_x$ denotes the
momentum of the particle at the site $x$ at time $t$.  The momentum is a
random variable whose value depends both on the realization of the flips and on
the distribution of the initial data at $t=0$.   We use
$\mean{\cdot}$ to denote the corresponding expectation values.  We prove
here that for a large class of
harmonic interactions with pinning, the average kinetic temperature
profile does thermalize and its
evolution will follow the time evolution dictated by the dynamic Fourier's law
as soon as a thermalization period $t_0 = O(L^{2/3})\ll L^2$ has passed.

As mentioned above, the velocity flip model has been studied before, using
several different methods from numerical to mathematically rigorous analysis.
In 
\cite{FFL94}, it was proven that every translation invariant stationary state of the
infinite chain  
with a finite entropy density is
given by a mixture of canonical Gibbs states, hence with temperature
as the sole parameter.   This was shown to hold even when fairly generic
anharmonic 
interactions are included.  Since the dynamics conserves total energy, this
provides 
strong support to the idea that energy is the only ergodic variable in the
velocity flip model with pinning.  

Results from numerical analysis of the velocity flip model are described in
\cite{DKL11}. There the covariance of the nonequilibrium steady state of a chain
with Langevin heat baths attached to both ends was analyzed, and it was observed
that the second order correlations in the steady state coincide with those of a
similar, albeit more strongly stochastic, model of particles coupled everywhere
to self-consistently chosen heat baths \cite{bll02}.  Hence, in its steady state
the stationary Fourier's law is satisfied with an explicit formula for the
thermal conductivity; the full mathematical treatment of the case without
pinning is given in  \cite{BO11}.  

It was later proven that, unlike the self-consistent heat bath model, the
velocity flip model satisfies also the dynamic Fourier's law.  This was
postulated in  \cite{bkll11,bkll11b}, based on earlier mathematical work on
similar models by Bernardin and Olla (see e.g.\ \cite{bo05,B07}), and it was
later proven by Simon in  \cite{Simon12}.  (Although the details are only
given for the case without pinning, it is mentioned in the Remark after Theorem
1.2 that the proofs can be adapted to include interactions with pinning.)
Also the structure of steady state correlations and energy fluctuations are
discussed in \cite{bkll11b}
with supporting numerical evidence presented in \cite{bkll11}.  For a more
general explanatory discussion about hydrodynamic fluctuation theory, we refer
to a recent preprint by Spohn \cite{spohn13}.

The strategy for proving the hydrodynamic limit in \cite{Simon12} was based on
relative entropy methods introduced by Yau \cite{yau91} and Varadhan
\cite{var93}; see also \cite{ovy93} and \cite{kipnis99} for more references and
details.  There one begins by assuming that the initial state is close to a
local thermal equilibrium (LTE) state, which allows for unique definition of the
initial profiles of the hydrodynamic fields.  One considers
the relative entropy \defem{density} (i.e., entropy divided by the volume) 
of the
state evolved up to time $t$ with respect to a local equilibrium state
constructed from the hydrodynamic fields at time $t$, and the goal is to show
that the entropy density  
approaches zero in the infinite volume limit.  

In the present work we improve on the result proven in \cite{Simon12} in two
ways.  We prove that it is not necessary to assume that the initial state is
close to an LTE state; indeed, our main theorem is applicable for arbitrary
deterministic initial data, including those in which just one of the sites
carries energy. Instead, we allow for an
initial thermalization period---infinitesimal on the diffusive time scale---and
only after the period has passed is the evolution of the temperature profile
shown to follow a continuum heat equation. In particle systems directly coupled
to diffusion processes similar results have been obtained before: for instance,
in the references \cite{Uchi94,ShLu95} a hydrodynamic limit is proven assuming
only convergence of initial data. However, as discussed in Remarks 3.5 of
\cite{ShLu95}, even assuming a convergence restricts the choice of initial
data but we do not need to do it here.

Secondly, we show that the Fourier's law has a version involving
a lattice diffusion kernel for which the temperature profile is well
approximated by its macroscopic value at every lattice site and at every time
after the themalization period. 
This improves on
the standard estimates which only imply that the \defem{macroscopic averages} of
the two profiles coincide in the limit $L\to \infty$.  As shown in
\cite{yau94}, it is sometimes possible to use averaging
over smaller regions, of diameter $O(L^{a})$, $0<a<1$, but
it is not
easy to see how relative entropy alone could be
used to control local microscopic properties of the solution.  
The precise statements
and assumptions for our main results are given in Theorem \ref{th:maintprof} and
Corollary
\ref{th:maincoroll} in Sec.~\ref{sec:Tprofile}.

However, in two respects our result is less informative than the one in
\cite{Simon12}.  Firstly, we only describe the evolution of the \defem{average}
temperature profile, whereas the relative entropy methods produce statements
which describe the hydrodynamic limit profiles \defem{in probability}.
Secondly, we
do not prove here that the full statistics can be locally approximated by
equilibrium measures, although the estimate for the temperature profile does
indicate that this should be the case.  The thermalization of the other degrees
of freedom is only briefly discussed in Sec.~\ref{sec:LTE} where we introduce
a local version of the dynamic replica method.

The paper is organized as follows: We recall the definition of the velocity flip
model in Sec.~\ref{sec:model} and introduce various related notations there. 
The first version of the new tool, called \defem{global dynamic replica method},
is described in Sec.~\ref{sec:global} where we also discuss how it might be
applied to prove global equilibration for this model.  This discussion, as well
as the one involving the \defem{local dynamic replica method} in
Sec.~\ref{sec:LTE}, 
is not completely mathematically rigorous, for instance, due to missing regularity assumptions.
We have also included a discussion about applications of the dynamic replica method to other models
in Sec.~\ref{sec:anharm}.  To illustrate the expected differences to the present case, we briefly summarize there
the changes occurring when the velocity flips are replaced by an anharmonic onsite potential.

The main mathematical
content is contained in
Sec.~\ref{sec:Tprofile} where the global replica equations are applied to
provide a rigorous analysis of the time evolution of the kinetic temperature
profile, with the above mentioned 
Theorem \ref{th:maintprof} and Corollary \ref{th:maincoroll} as the main goals
of the section. 
We have included some related but previously known material in two Appendices. 
Appendix \ref{sec:applocal} concerns the explicit solution of the local dynamic
semigroup, and in Appendix \ref{sec:diffkernel} we derive the main properties of
the Green's function solution of the renewal equation describing the evolution
of the temperature profile.

\section{Velocity flip model on a circle}
\label{sec:model}

Mainly for notational simplicity, we consider here only one-dimensional periodic
crystals, i.e., particles on a circle.  For $L$ particles
we parametrize the sites on the circle by
\begin{align}
  & \Lambda_L := \Bigl\{-\frac{L-1}{2},\ldots,\frac{L-1}{2}\Bigr\} \, , \qquad
\text{if $L$ is odd},\\
  &\Lambda_L := \Bigl\{-\frac{L}{2}+1,\ldots,\frac{L}{2}\Bigr\}\, , \qquad
\text{if $L$ is even}.
\end{align}
Then always $|\Lambda_L|=L $ and $\Lambda_L\subset \Lambda_{L'}$ if $L\le L'$.
In addition, 
for odd $L$, we have $\Lambda_L = \defset{n\in \Z }{|n|<\frac{L}{2}}$.  We use
periodic arithmetic on $\Lambda_L$, setting 
$x'+x := (x'+x) \bmod \Lambda_L$ for $x',x\in  \Lambda_L$.  Sometimes we will
need lattices of several different
sizes simultaneously, and to stress the length of the cyclic group, we then
employ the notation $[x'+x]_L$ for $x'+x$.  Also, we use $-x$ to denote
$[0-x]_L$.

The particles are assumed to interact via linear forces with a finite range. 
The forces are determined by a map $\Phi:\Z \to \R$ which we assume to be
symmetric, $\Phi(-x)=\Phi(x)$.  We choose $\rphi$ to be odd and assume that
$\Phi(x)=0$ for all $|x| \ge \rphi/2$.  Then the support of $\Phi$ lies in
$\Lambda_{\rphi}$.  The forces are assumed to be stable and pinning, i.e., 
the discrete Fourier transform of $\Phi$ is required to be strictly positive. 
The square root
of the Fourier transform determines the  dispersion relation $\omega:\T \to \R$
which is a smooth function on the circle $\T:=\R/\Z$
with $\omega_0 := \min_{k\in \T} \omega(k) > 0 $.
We define the corresponding periodic interaction matrices $\Phi_L \in
\R^{\Lambda_L\times \Lambda_L}$ on $\Lambda_L$ by setting
\begin{align}\label{eq:defphiL}
  (\Phi_L)_{x',x} := \Phi([x'-x]_L)\, , \quad \text{for all }x',x\in 
\Lambda_L\, .
\end{align}
This clearly results in a real symmetric matrix.

Fourier transform $\mathcal{F}_L$ maps functions
$f:\Lambda_L\to\C$ to $\FT{f} : \Lambda_L^* \to
\C$, where $\Lambda^* := \Lambda_L/L \subset (-\frac{1}{2},\frac{1}{2}]$ is the
dual lattice and for $k\in \Lambda_L^*$ we set
\begin{equation} 
   \FT{f}(k) = \sum_{x\in \Lambda_L} f(x) \textrm{e}^{-\textrm{i} 2\pi k \cdot
x} \, .
   \end{equation} 
The formula holds in fact for all $k\in \mathbb{Z}/L$, in the sense that the
right hand side is then equal to $\FT{f}(k \bmod \Lambda_L^*)$, i.e., it
coincides with the periodic extension of $\FT{f}$.
The inverse transform $\mathcal{F}_L^{-1}:g\mapsto \tilde{g}$ is given by
\begin{equation} 
   \tilde{g}(x) = \int_{\Lambda_L^*} \textrm{d} k\,g(k)
     \textrm{e}^{\textrm{i} 2\pi k \cdot x}\,,
\end{equation}
where we use the convenient shorthand notation
\begin{equation} 
  \int_{\Lambda_L^*} \textrm{d} k\, \cdots =
   \frac{1}{|\Lambda_L|} \sum_{k\in \Lambda_L^*} \cdots \, .
\end{equation}
With the above conventions, for any $L\ge \rphi$ we have
\begin{align}
  (\mathcal{F}_L \Phi_L \mathcal{F}_L^{-1})_{k',k} := \omega(k)^2
\delta_L(k'-k)\,
, \quad \text{for all }k',k\in  \Lambda_L^*\, ,
\end{align}
where $\delta_L$ is a ``discrete $\delta$-function'' on $\Lambda^*_L$,
defined by 
\begin{equation}
  \delta_{L}(k) = |\Lambda_L| \mathbbm{1}(k = 0) \, , \quad \text{for }k\in 
\Lambda^*_L\, .
\end{equation}
Here, and in the following, $\cf$ 
denotes the generic characteristic function: $\cf(P)=1$ if the condition $P$ is
true, and otherwise $\cf(P)=0$.

We assume all particles to have the same mass, and choose units in which the
mass is equal to one.
The linear forces on the circle are then generated by the Hamiltonian
\begin{align}\label{eq:defHLandGL}
  & H_L(X)  := \sum_{x\in \Lambda_L} \frac{1}{2} (X_x^2)^2 
  + \sum_{x',x\in \Lambda_L} \frac{1}{2} X_{x'}^1 X_{x}^1 \Phi([x'-x]_L)
  = \frac{1}{2} X^T \mathcal{G}_L X\, ,\\
  & \mathcal{G}_L := \begin{pmatrix}
                    \Phi_L & 0 \\ 
                    0 & 1
                  \end{pmatrix} \in \R^{(2 \Lambda_L)\times(2 \Lambda_L)} \, ,
\end{align}
on the phase space $X\in \Omega:= \R^{\Lambda_L}\times \R^{\Lambda_L}$. 
The canonical pair of variables for the site $x$ are
the position $q_x := X^1_x$, and the momentum $p_x := X^2_x$.  
The Hamiltonian evolution is combined with a velocity-flip noise.  
The resulting system can be identified with a Markov process $X(t)$ and
the process generates a Feller semigroup on the space of observables vanishing
at infinity, see \cite{BO11,Simon12} for mathematical details.  
Then for $t> 0$ and any $F$ in the domain of the generator $\mathcal{L}$ of the
Feller process the expectation values of $F(X(t))$ satisfy an evolution
equation
\begin{align}\label{eq:mainevoleq}
  &\partial_t \mean{F(X(t))} = \mean{(\mathcal{L} F)(X(t))},
    \\ \intertext{where $\mathcal{L} := \mathcal{A} + \mathcal{S}$, with}
  &\mathcal{A} := 
  \sum_{x\in \Lambda_L} \left(X_x^2 \partial_{X_x^1} - (\Phi_L X^1)_x
\partial_{X_x^2}\right)\, ,\\  
  &(\mathcal{S}F)(X) := \frac{\gamma}{2} \sum_{x_0\in \Lambda_L} \left(
F(S_{x_0} X)- F(X) \right), \quad \gamma>0\, ,\\
  &(S_{x_0} X)_x^i := \begin{cases}
                 -X_x^i\, , & \text{if }i=2 \text{ and }x=x_0\, , \\
                 X_x^i\, , & \text{otherwise}\, .
               \end{cases}
\end{align}

We consider the time evolution of the moment generating function
\begin{align}
  f_t(\xi) := \mean{\rme^{\ci \xi\cdot X(t)}}\, ,
\end{align}
where $\xi$ belongs to some fixed neighborhood of $0$.  Although the observable
$X \mapsto \rme^{\ci\xi\cdot X}$ does not vanish at infinity, $f_t(\xi)$ is
always well defined, and
we assume that it satisfies the evolution equation
dictated by (\ref{eq:mainevoleq}).  This will require some additional
constraints on the distribution of initial data, but for instance it should
suffice that all second moments of $X(0)$ are finite. (The existence of
initial second moments will also be an assumption for our main theorem.)

Ultimately, the goal is to prove thermalization,
i.e., the appearance of local thermal equilibrium.  More precisely,
we would like to prove that after a thermalization period the local restrictions
of the generating functional are well approximated by mixtures of equilibrium
expectations.  To do such a comparison, the first step is to classify the
generating functions of equilibrium states.  We start with a heuristic
argument based on ergodicity which gives a particularly appealing formulation
for the present case in which the canonical Gibbs states are Gaussian measures.

Suppose
that the evolution of our finite system is ergodic, with energy as the only
ergodic variable.  Then for any invariant measure $\bar{\mu}$ there is a Borel
probability measure $\nu$ on $\R$ such that for any $g\in L^1(\bar{\mu})$ we
have 
\begin{align}
  \int\! \bar{\mu}(\rmd X) g(X) = \int\! \nu(\rmd E) \int\!  \frac{\rmd
X}{Z^{\text{mc}}_E}\delta(E-H_L(X)) g(X)\, ,
\end{align}
where $Z^{\text{mc}}_E:= \int \! \rmd X \delta(E-H_L(X))$ denotes the
microcanonical partition function.  (Details about mathematical ergodic theory
can be found for instance from \cite{EW11}.) 
Hence, if $\mu_0$ is an  initial state
which converges towards a steady state $\bar{\mu}$, we have for any
observable $g$
\begin{align}
&  \lim_{t\to\infty} \int\! \mu_t(\rmd X) g(X) = \int\! \nu(\rmd E) \int\! 
\frac{\rmd X}{Z^{\text{mc}}_E}\delta(E-H_L(X)) g(X)\, .
\end{align}
Applying this for $g(X) = \varphi(H_L(X))$, $\varphi$ continuous with a compact
support, we find by conservation of $H_L$ that 
$\int\! \mu_0(\rmd X) \varphi(H_L(X)) = \int\! \mu_t(\rmd X) \varphi(H_L(X)) =
\int\! \nu(\rmd E) \varphi(E)$.  Therefore, we can formally identify
$\nu(\rmd E)=  \rmd E \int\! \mu_0(\rmd X)  \delta(E-H_L(X))$.

Finally, we can rewrite the somewhat unwieldy microcanonical expectations in
terms of the canonical Gaussian measures by using the representation
$\int\!\rmd X\,\delta(E-H_L(X)) g(X)= \int_{\beta-\ci \infty}^{\beta+\ci \infty}
\frac{\rmd z}{2\pi \ci} \rme^{z E}
\int\!\rmd X\,\rme^{-z H_L(X)} g(X)$ which should be valid for all
sufficiently nice $g$ 
and $\beta>0$.
Applying this representation to $g(X)=\rme^{\ci \xi\cdot X}$ yields
$\int_{\beta-\ci \infty}^{\beta+\ci \infty} \frac{\rmd z}{2\pi \ci} \rme^{z E}
Z^{\text{can}}_z \rme^{-\frac{1}{2 z} \xi^T \mathcal{G}_L^{-1} \xi}$, where
the canonical partition function is $Z^{\text{can}}_z := \int\!\rmd X\,\rme^{-z
H_L(X)}$.  Hence, we arrive at the conjecture that for all sufficiently nice
initial data $\mu_0$
\begin{align}\label{eq:ftlim}
& \lim_{t\to\infty} f_t(\xi) = \int_{\beta-\ci \infty}^{\beta+\ci \infty}
\frac{\rmd z}{2\pi \ci} 
\rme^{-\frac{1}{2 z} \xi^T \mathcal{G}_L^{-1} \xi} \int\! \mu_0(\rmd X)
\frac{\int\!\rmd X' \rme^{z(H_L(X)-H_L(X'))}}{\int\!\rmd X'
\delta(H_L(X)-H_L(X'))}
 \, .
\end{align}
By a change of variables to $z^{-1}$ and using the fact that $f_t(0)=1$, the
limit function can also be represented in the form 
$\oint\! \bar{\nu}(\rmd \lambda)\, \rme^{-\lambda \frac{1}{2} \xi^T
\mathcal{G}_L^{-1} \xi}$ where the integral is taken around a circle in the
right half of the complex plane and $\bar{\nu}$ is a complex measure
satisfying $\oint\!\bar{\nu}(\rmd \lambda)=1$.

For any fixed initial data $X_0$ with energy $E=H_L(X_0)$, there is a
natural choice for the parameter $\beta$ as the unique
solution to the equation $E=\mean{H_L}_\beta^{\text{can}}$.  This choice
coincides with 
the unique saddle point for $\xi=0$ on the positive real axis, i.e., it is the
only $\beta>0$ 
for which $\partial_z \ln g(z)|_{z=\beta}=0$ with $g(z) := \int\!\rmd X'
\rme^{z(E-H_L(X'))}/\int\!\rmd X' \delta(E-H_L(X'))$.  Then also $\partial_z^2
\ln g(\beta)= \mathrm{Var}_\beta(H_L)>0$ and hence the integration path in
(\ref{eq:ftlim}) follows the path of steepest descent through the saddle point.
As the energy variance typically is proportional to the volume, the integrand
should be
concentrated to the real axis, with a standard deviation $O(L^{1/2})$.  Hence,
for 
fixed initial data and large $L$ we would expect to have here equivalence of
ensembles in the form $\lim_{t\to\infty} f_t(\xi) \approx \rme^{-\frac{1}{2
\beta} \xi^T \mathcal{G}_L^{-1} \xi}$.

\section{Global dynamic replica method}
\label{sec:global}

In order to treat the local dynamics independently, we replicate the whole
chain at each 
lattice site, and transform the evolution equation into a new form by selecting
some terms to act on the replicated direction.  
We use a generating function with variables $\zeta \in \R^{2\times
\Lambda_L\times \Lambda_L}$, where each $\zeta^i_{x,y}$ controls the random
variable $X^i_{x+y}(t)$, and we think of $x$ as the original site
and $y$ as the position in its ``replica''.  Explicitly, we study the dynamics
of the generating function 
\begin{align}
  h_t(\zeta) := \mean{\rme^{\ci \sum_{x,y,i} \zeta^i_{x,y} X(t)^i_{x+y}}}
\end{align}
where the mean is taken over the distribution of $X(t)$ at time $t$ for some
given initial distribution $\mu_0$ of $X(0)$.  Clearly, $h_t$ depends on
$\zeta$ only via the combinations $\sum_y \zeta^i_{x-y,y}$, $x\in \Lambda_L$.
If $h_t$ is known, the
local
statistics at $x_0\in \Lambda_L$ for some given time $t$ can be obtained
directly from its restriction $f_{t,x_0}(\xi):= h_t(\zeta[\xi,x_0])$ where we
set $\zeta[\xi,x_0]_{xy}^i := \cf(x{=}x_0,\,y{\in} \Lambda_{R_0}) \xi^i_y$ for
$\xi \in \R^{2\times \Lambda_{R_0}}$.    We assume $R_0\le L$ but otherwise it can be chosen
independently of $L$.   The parameter $R_0$
determines which neighboring particles are chosen to belong to the same
``local'' 
neighborhood.  

By (\ref{eq:mainevoleq}), the generating function $h_t$ satisfies the evolution
equation
\begin{align}\label{eq:htevol1}
  & \partial_t h_t(\zeta) = 
  \sum_{x,y\in \Lambda_L} \zeta^1_{xy} \mean{\ci X^2_{x+y} \rme^{\ci
Y_t}}-\sum_{x,y,z\in \Lambda_L} \zeta^2_{xy} 
 (\Phi_L)_{x+y,x+z}\mean{\ci X^1_{x+z} \rme^{\ci Y_t}} 
      \nonumber \\ & \quad
      + \frac{\gamma}{2} \sum_{x_0\in \Lambda_L} \left( h_t(\sigma_{x_0}
\zeta)-h_t(\zeta)\right)
\, ,
\end{align}
where we use the random variable $Y_t := \sum_{x,y,i} \zeta^i_{x,y}
X(t)^i_{x+y}$ and have defined for $x_0\in \Lambda_L$
\begin{align}
  (\sigma_{x_0} \zeta)_{xy}^i := \begin{cases}
                 -\zeta_{xy}^i\, , & \text{if }i=2 \text{ and }[x+y]_L=x_0\, ,
\\
                 \zeta_{xy}^i\, , & \text{otherwise}\, .
               \end{cases}
\end{align}
The equation can  be closed by using the identity
\begin{align}\label{eq:tozetader}
  & \partial_{\zeta_{xy}^i} h_t(\zeta) = \mean{\ci X^i_{x+y} \rme^{\ci Y_t}}\, .
\end{align}
There is some arbitrariness in the resulting equation: (\ref{eq:tozetader}) is
true for
all $x,y\in \Lambda_L$, but the right hand side depends only on $x+y$.  We
choose 
here to use it as indicated by the choice of summation variables in
(\ref{eq:htevol1}).  Since in the summand always $(\Phi_L)_{x+y,x+z} =
(\Phi_L)_{yz}$ this results in the evolution equation 
\begin{align}\label{eq:htevol2}
  & \partial_t h_t(\zeta) = - (\mathcal{M}_0 \zeta)\cdot \nabla h_t  (\zeta)    
 + \frac{\gamma}{2} \sum_{x_0\in \Lambda_L} \left( h_t(\sigma_{x_0}
\zeta)-h_t(\zeta)\right)
\, ,
\end{align}
where $\nabla h$ denotes the standard gradient, i.e., it is a vector whose
$(i,x,y)$-component is $\partial_{\zeta_{xy}^i} h$, and
\begin{align}\label{eq:defMM}
  \mathcal{M}_\gamma := \bigoplus_{x_0\in \Lambda_L} M_\gamma^{(x_0)}\, , \quad
  (M_\gamma^{(x_0)}\zeta)_{xy}^i  := \cf(x{=}x_0)\,
(M_\gamma\zeta_{x_0\cdot}^{\cdot})_{y}^i\, , \quad
  M_\gamma := \begin{pmatrix}
                0 & \Phi_L \\
                -1 & \gamma 1
              \end{pmatrix} \, .
\end{align}
Since $\frac{1}{2} \sum_x (1-\sigma_x) = P^{(2)}:= \begin{pmatrix}
                0 & 0 \\
                0 & 1
              \end{pmatrix}$, 
this can also be written as 
\begin{align}
  & \partial_t h_t(\zeta) = 
  - (\mathcal{M}_\gamma \zeta)\cdot \nabla h_t  (\zeta) 
 + \frac{\gamma}{2} \sum_{x_0\in \Lambda_L} 
  \left( h_t(\sigma_{x_0} \zeta)-h_t(\zeta)-(\sigma_{x_0}-1) \zeta\cdot \nabla
h_t  (\zeta)\right)
   \,.  
\end{align}
For any $h_t$ resulting from the replication procedure, we obviously have
$\partial_{\zeta^i_{xy}} h_t(\zeta) = \partial_{\zeta^i_{x'y'}}  h_t(\zeta)$
whenever $x+y=x'+y'$. 
Hence, by Taylor expansion with remainder up to second order 
\begin{align}
  & \frac{\gamma}{2} \sum_{x_0\in \Lambda_L} 
  \left( h_t(\sigma_{x_0} \zeta)-h_t(\zeta)-(\sigma_{x_0}-1) \zeta\cdot \nabla
h_t  (\zeta)\right)
    \nonumber \\ & \quad
  = 2 \gamma \sum_{x_0\in \Lambda_L} \int_0^1\! \rmd r\, (1-r) 
   \sum_{x'y'xy} \cf(x_0=x+y) \cf(x_0=x'+y')
    \nonumber \\ & \qquad \times
   \zeta^2_{xy} \zeta^2_{x'y'}  
   (\partial_{\zeta^2_{xy}}\partial_{\zeta^2_{x'y'}}
h_t)(\zeta-r(1-\sigma_{x_0})\zeta)
       \nonumber \\ & \quad
=  2 \gamma \int_0^1\! \rmd r\, (1-r) 
\sum_{x_0\in \Lambda_L} (\partial^2_{\zeta^2_{x_0,0}}
h_t)(\zeta-r(1-\sigma_{x_0}) \zeta)\, 
 \Bigl(\sum_{xy}\cf(x_0=x+y) \zeta^2_{xy}\Bigr)^2 \,.
\end{align}

Now for any continuously differentiable function $t\mapsto \zeta_t$
\begin{align}
  h_t(\zeta_0)-h_0(\zeta_t) = -\int_0^t\! \rmd s\, \partial_s (h_{t-s}(\zeta_s))
  = \int_0^t\! \rmd s\,\left(\dot{h}_{t-s}(\zeta_s)-\dot{\zeta}_s\cdot\nabla
h_{t-s}(\zeta_s)\right)\, .
\end{align}
Setting 
$\zeta_s := \rme^{-s \mathcal{M}_\gamma} \zeta$ and $Q_{r,x_0} :=
1-r(1-\sigma_{x_0})$ thus yields the ``Duhamel formula''
\begin{align}\label{eq:htiter}
  & h_t(\zeta)=h_0(\zeta_t)+\int_0^t\! \rmd s\,
  \frac{\gamma}{2} \sum_{x_0\in \Lambda_L} 
  \left( h_{t-s}(\sigma_{x_0} \zeta_s)-h_{t-s}(\zeta_s)-(\sigma_{x_0}-1)
\zeta_s\cdot \nabla h_{t-s}  (\zeta_s)\right)
      \nonumber \\ & \quad
  = h_0(\zeta_t)+2 \gamma\int_0^t\! \rmd s \int_0^1\! \rmd r\, (1-r) 
\sum_{x_0\in \Lambda_L} 
 \Bigl(\sum_{y\in \Lambda_L} 
(\zeta_s)^2_{x_0-y,y}\Bigr)^2 (\partial^2_{\zeta^2_{x_0,0}}
h_{t-s})( Q_{r,x_0}\zeta_s)
 \,.
\end{align}
In this formula, the replica dynamics has been exponentiated in the
operator semigroup $\rme^{-s \mathcal{M}_\gamma}$.  No approximations have been
made in the derivation of the formula, but to show that
its solutions, under some natural assumptions, are unique and correspond to LTE
states does not look straightforward. We do
not attempt to do it here.  Instead, the formula will be used in the next
section to derive a closed evolution equation for the temperature profile.

We conclude the section by showing that Eq.~(\ref{eq:htiter}) is consistent
with the discussion in Sec.~\ref{sec:model}.
Suppose $\nu$ is a complex bounded measure on $\beta+\ci \R$, for some
$\beta>0$.
Then $h_t(\zeta) := \int\!\nu(\rmd z)\,
\rme^{-\frac{1}{2 z} \xi^T \mathcal{G}_L^{-1} \xi}$, with $\xi^i_x = \sum_y
\zeta^i_{x-y,y}$, solves (\ref{eq:htiter}).  To see this, first note that the
first two terms in the integrand cancel, since $(\sum_y
(\sigma_{x_0}\zeta)^i_{x-y,y})^2 = (\sum_y \zeta^i_{x-y,y})^2$ and thus
$h_{t}(\sigma_{x_0} \zeta)=h_{t}(\zeta)$.
Therefore, the value of the integral is equal to $\int_0^t\! \rmd s\, \gamma
P^{(2)}\zeta_s\cdot \nabla h_{0}  (\zeta_s) $.  Here
$P^{(2)}\zeta_s\cdot \nabla h_{0}  (\zeta_s) = -\int\!\nu(\rmd z)\, \frac{1}{z}
\rme^{-\frac{1}{2 z} \xi_s^T \mathcal{G}_L^{-1} \xi_s} 
\sum_{x} ((\xi_s)^2_x)^2$, with $(\xi_s)^i_x = \sum_{i',y',y} (\rme^{-s
M_\gamma})^{i i'}_{y y'}\zeta_{x-y,y'}^{i'}=(\rme^{-s M_\gamma} \xi)^i_x$ where
in the second equality we have used the periodicity of $M_\gamma$.  However,
then
$\partial_s \left(\xi_s^T \mathcal{G}_L^{-1} \xi_s\right)=-\xi_s^T (M_\gamma^T
\mathcal{G}_L^{-1}+ \mathcal{G}_L^{-1} M_\gamma)\xi_s = -2 \gamma \sum_x
((\xi_s)^2_x)^2$, and thus 
$\int_0^t\! \rmd s\, \gamma P^{(2)}\zeta_s\cdot \nabla h_{0}  (\zeta_s)
=-\int_0^t\! \rmd s\, \partial_s h_0(\zeta_s)=-h_0(\zeta_t)+h_t(\zeta)$.
Hence, the functions of $\zeta$ defined by setting $\xi^i_{x}=\sum_y
\zeta^i_{x-y,y}$ on the right hand side of (\ref{eq:ftlim}) are solutions to the
equation (\ref{eq:htiter}). To
check that energy is the only ergodic variable one would need to prove
that there are no other time-independent solutions.  
We postpone the analysis of this question to a
future work, although by the results proven in \cite{FFL94} it would seem to be
a plausible conjecture.

\section{Thermalization of the temperature profile}
\label{sec:Tprofile}

Since the ``replicated'' generating function satisfies (\ref{eq:htiter}), a
direct differentiation results in an evolution equation for the kinetic
temperature profile $T_{t,x} := \mean{(X(t)^2_x)^2} = -
\partial^2_{\zeta^2_{x0}} h_t(0)$.  We obtain
\begin{align}
& T_{t,x} =  (\rme^{-t M_\gamma^T} \Gamma_x \rme^{-t M_\gamma})^{22}_{00} 
+ 2 \gamma \int_0^t\! \rmd s\,
  \sum_{y\in \Lambda_L}  ((\rme^{-s M_\gamma})^{2 2}_{y0})^2 T_{t-s,x+y}
\, ,
\end{align}
where each
$(\Gamma_x)^{i'i}_{y'y}:=-\partial_{\zeta^{i'}_{xy'}}\partial_{\zeta^i_{xy}}
h_0(0) = \mean{X(0)^{i'}_{x+y'} X(0)^i_{x+y}}$ is a symmetric matrix obtained by
a periodic translation of the matrix of the initial second moments. 
The final sum can be transformed into a standard convolution form by
changing the summation variable $y\to -y$.  This yields 
\begin{align}\label{eq:Tevoleq}
& T_{t,x} =  g_{t,x} + \int_0^t\! \rmd s\, 
  \sum_{y\in \Lambda_L} p_{s,y} T_{t-s,x-y}
\, ,
\end{align}
where for $t\ge 0$, $x\in \Lambda_L$, the ``source term'' $g$ and the ``memory kernel'' $p$ are given by
\begin{align}
& g_{t,x} := (\rme^{-t M_\gamma^T} \Gamma_x \rme^{-t M_\gamma})^{22}_{00} =
\sum_{i'iy'y}  A^i_{t,y}  A^{i'}_{t,y'}
(\Gamma_x)^{i'i}_{y'y}\, ,
  \\ &
p_{t,x} := 2 \gamma (A^2_{t,-x})^2\, ,
\end{align}
with the following shorthand
\begin{align}
  A^i_{t,x}:= (\rme^{-t M_\gamma})^{i2}_{x0}\, .
\end{align}

We will prove later that 
$p_{t,-x}=p_{t,x}\ge 0$ and that $\int_0^\infty \!\rmd t \sum_x p_{t,x}=1$.
Hence, mathematically the equation (\ref{eq:Tevoleq}) has the structure of a
generalized \defem{renewal equation}.
Renewal equations have bounded solutions in great generality 
\cite[Theorem~9.15]{Rudin:FA}.  The problem is closely connected to Tauberian
theory; the classical paper by Karlin \cite{Karlin55} contains a discussion and
detailed analysis of the standard case.
Unfortunately, most of these results are not of direct use here since
they do not give estimates for the speed of convergence towards the asymptotic
value and thus cannot be used for estimating the
$L$-dependence of the asymptotics.  Nevertheless, the standard methods can be
applied to an extent also in the present case: in Appendix
\ref{sec:diffkernel} we give the details for the existence and uniqueness
of solutions to (\ref{eq:Tevoleq}) and 
derive an explicit representation of the solutions using Laplace transforms.

The analysis relies on upper and lower bounds for the tail behavior of
$p_{t,x}$.  These follow from explicit formulae for the solutions of the
semigroup $\rme^{-t M_\gamma}$ derived in Appendix \ref{sec:applocal}.  
In particular, we have for any $k\in \Lambda^*_L$
\begin{align}\label{eq:Aiexp}
&  \FT{A}^1_t(k) = \frac{1}{2 u(k)} \sum_{\sigma=\pm 1} \sigma \omega(k)^2
\rme^{-t \mu_\sigma(k)}\, ,\quad
  \FT{A}^2_t(k) = \frac{1}{2 u(k)} \sum_{\sigma=\pm 1} \sigma \mu_\sigma(k)
\rme^{-t \mu_\sigma(k)}\, ,
\end{align}
where
\begin{align}\label{eq:defmupm}
  u(k) := \sqrt{(\gamma/2)^2-\omega(k)^2}\, ,\quad
  \mu_\sigma(k) := \frac{\gamma}{2}+ \sigma u(k)\, .
\end{align}
In principle, the formulae should only be used if $\omega(k)< \gamma/2$ which
implies $u(k)> 0$.  However, they also hold for all other values of $k$ if
extended using the following ``analytic continuation'':  if $\omega(k)>
\gamma/2$, we set $u(k) = \ci \sqrt{\omega(k)^2-(\gamma/2)^2}$ and the values
for case $\omega(k)=\gamma/2$ agree with the limit $u(k)\to 0^+$.
Since we
consider  here the case in which the noise dominates,  only the expressions in
(\ref{eq:Aiexp}) will be needed in the following.

We begin by summarizing the regularity assumptions about the free evolution,
already discussed in Sec.~\ref{sec:model}.  Without additional effort, we can
relax the assumption of $\Phi$ having a finite support to mere exponential
decay. There are then several possibilities for fixing the finite volume
dynamics; here, we set $\omega(k;L):= \sqrt{\FT{\Phi}(k)}$ for $k\in
\Lambda_L^*$. Then (\ref{eq:Aiexp}) is still pointwise valid for the Fourier
transform of the semigroup generated by $M_\gamma$.
\begin{assumption}\label{th:phiassump}
  We assume that the map $\Phi:\Z\to \R$ has all of the following properties.
  \begin{enumerate}
  \setlength{\itemsep}{0pt}
    \item {\em (exponential decay)}\/ There are $C,\delta>0$ such $|\Phi(x)|\le
C \rme^{-\delta |x|}$ for all $x\in \Z$.
    \item {\em (symmetry)}\/ $\Phi(-x)=\Phi(x)$  for all $x\in \Z$.
    \item {\em (pinning)}\/ There is $\omega_0>0$ such that $\FT{\Phi}(k)\ge
\omega_0^2$ for all $k\in \T$.
  \end{enumerate}
\end{assumption}
The assumptions imply that the Fourier transform of $\Phi$ can be extended to an
analytic map $z\mapsto \sum_{x} \rme^{-\ci 2\pi z x} \Phi(x)$ on the strip
$\R+\ci (-\delta,\delta)/(2\pi)$, and the extension is $1$-periodic.  By
continuity and periodicity of $\FT{\Phi}$, we can then find $\vep>0$ 
such that
$\re \FT{\Phi}(z)>0$ on the strip $\R+\ci (-\vep,\vep)$.  Therefore, the
infinite volume dispersion relation has the following regularity properties:
\begin{corollary}
Assume $\Phi$ satisfies the assumptions in \ref{th:phiassump}. Then
$\omega(k):=\sqrt{\FT{\Phi}(k)}$, $k\in \T$, defines a smooth function on $\T$
which satisfies $\omega(k)\ge \omega_0$ and $\omega(-k)=\omega(k)$ for all $k\in
\T$.  In addition, there is $\vep>0$ such that $\omega$ has an analytic,
$1$-periodic continuation to the region $\R+\ci (-\vep,\vep)$.
\end{corollary}

From now on we assume that $\Phi$ satisfies Assumption \ref{th:phiassump} and
$\gamma>0$ is some fixed flipping rate.  This already fixes the functions
$\FT{A}^i$ defined above.  However, here we aim at convenient exponential bounds
for the errors from diffusive evolution of the temperature profile.  This
requires to rule out resonant behavior, which can be achieved if the noise
flipping rate is high enough and the dispersion relation satisfies a certain
integral bound excluding degenerate behavior.
Explicitly, we only consider $\gamma$ and $\Phi$ satisfying the following:
\begin{assumption}\label{th:gammaassump}
Suppose that $\Phi$ satisfies the assumptions in \ref{th:phiassump}, and
$\gamma>0$ is given such that:
  \begin{enumerate}
  \setlength{\itemsep}{0pt}
    \item {\em (noise dominates)}\/ $\gamma^2>4 \max_k \FT{\Phi}(k)$.
    \item {\em (harmonic forces are nondegenerate)}\/ For any $\vep>0$ there is
$C_\vep>0$ such that
    \begin{align}\label{eq:nondegcond}
      \int_{0}^\infty\!\rmd t \int_{\T}\! \rmd k \,
\left(\FT{A}_t^2\Bigl(k+\frac{k_0}{2}\Bigr)-\FT{A}_t^2\Bigl(k-\frac{k_0}{2}
\Bigr)\right)^2 \ge C_\vep\, ,\quad
      \text{whenever } \vep\le |k_0|\le \frac{1}{2}\, .
    \end{align}
  \end{enumerate}
\end{assumption}
The nondegeneracy condition is 
satisfied by the nearest neighbor interactions, for which $\omega(k) =
\sqrt{\omega_0^2+ 4 \sin^2(\pi k)}$ with $\omega_0>0$.  This can be proven for
instance by relying on the estimate $\partial_k \FT{A}^2_t(k)\ge C t \rme^{-t
\gamma} \sin(2\pi k)$ valid for all $|k|<1/2$ and large enough $t$.  However,
the condition fails for the degenerate next-to-nearest neighbor coupling which
skips over the nearest neighbors: then $\omega(k) = \sqrt{\omega_0^2+ 4 \sin^2(2
\pi k)}$ and thus for $k_0=1/2$ we have $\omega(k+k_0/2)=\omega(k-k_0/2)$ for
all $k$, hence the integral in (\ref{eq:nondegcond}) evaluates to zero at
$k_0=1/2$.  Instead of including a formal proof of these statements, we have
depicted the values of the above integrals for one choice of parameters in
Fig.~\ref{fig:degintplots}.

\begin{figure}
\centering
\myfigure{width=0.48\textwidth}{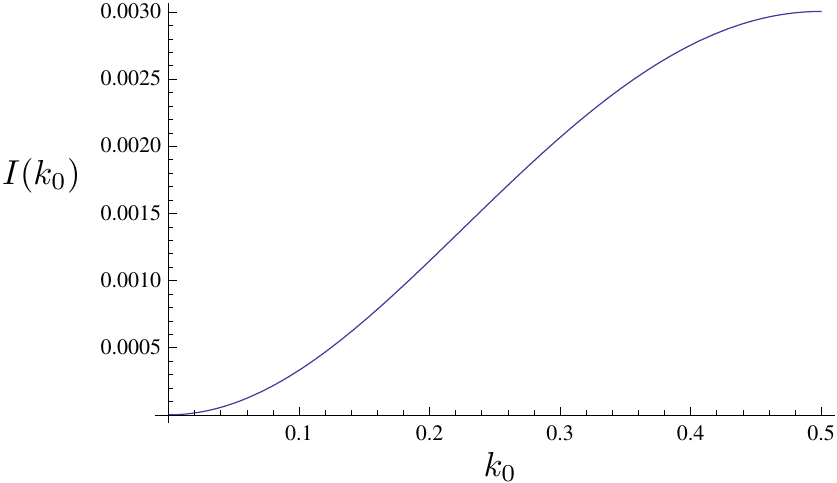}\ 
\myfigure{width=0.48\textwidth}{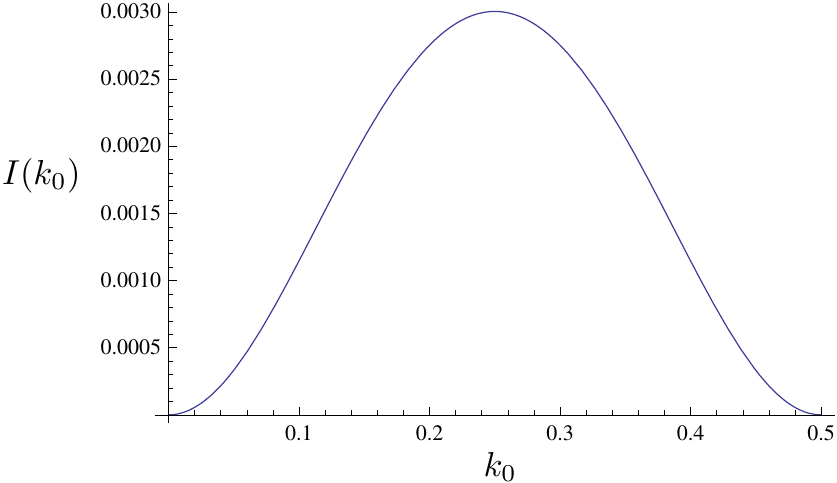}
  \caption{Plots of the function $I(k_0)$ defined by the left hand side of the
    nondegeneracy 
condition in (\ref{eq:nondegcond}) for $\gamma=6$ and two different dispersion
relations.  The left plot is for the standard nearest neighbor case, $\omega(k)
= \sqrt{1+ 4 \sin^2(\pi k)}$, while the right one depicts a degenerate
next-to-nearest 
neighbor case, with $\omega(k) = \sqrt{1+ 4 \sin^2(2 \pi k)}$.  The plots have
been generated by numerical integration using
Mathematica.\label{fig:degintplots}}
\end{figure}

We prove in this section that these assumptions suffice to have the following
pointwise behavior of the temperature profile.  
\begin{theorem}\label{th:maintprof}
  Suppose $\Phi$ and $\gamma$ satisfy the conditions in Assumption
\ref{th:gammaassump}.  Then there is $L_0\in \N$ such that 
equation (\ref{eq:Tevoleq}) has a unique 
continuous solution $T_{t,x}$ for every $L\ge L_0$
whenever all second moments of the initial field $X(0)$ exist. 
Let $\mathcal{E}:=
|\Lambda_L|^{-1}\mean{H_L(X)}<\infty$ denote the total energy density.
Then there are constants $C,d>0$, independent of the initial data and of $L$,
such that for this solution
\begin{align}\label{eq:TdiffersfromE}
\left|T_{t,x}-\mathcal{E}\right|\le C\mathcal{E}
 \frac{ \rme^{-d t L^{-2}} }{1-\rme^{-2 d t L^{-2}}}\, ,
\end{align}
for all $t\ge 0$ and $x\in\Lambda_L$.  In addition, we can choose $C$ and
define 
$\tau_x\in \R$ and $\tilde{p}_x\ge 0$, $x\in \Lambda_L$, so that for all $t>0$
and $x\in \Lambda_L$ 
\begin{align}\label{eq:Tlatticediff}
\left|T_{t,x}-(\rme^{-t D}\tau)_x\right|\le C\mathcal{E} L t^{-3/2}\, ,
\end{align}
where the operator $D$ is defined by 
\begin{align}
  (D \tau)_x := \sum_{y\in \Lambda_L} \tilde{p}_y  \left(2 \tau_x- \tau_{x+y} -
\tau_{x-y} \right)\, , \quad x\in\Lambda_L\, .
\end{align}
\end{theorem}
Both $\tau_x$ and $\tilde{p}_x$ in the statement have explicit
definitions which can be found in the beginning of the proof of the
Theorem.  To summarize in words, the first of the bounds implies that the
temperature profile equilibrates, and the relative error is exponentially
decaying on the diffusive time scale, i.e., as $t L^{-2}$ becomes large.  The
second 
statement says that solving the ``lattice diffusion equation'' $\partial T_{t,x}
= - (DT_t)_x$ with initial data $T_0=\tau$ provides an approximation to the
temperature profile which is accurate even before the diffusive time scale, for
$t\gtrsim L^{2/3}$.  

A closer inspection of the proof of the Theorem reveals that
the main contribution to the error bound given in (\ref{eq:Tlatticediff}) comes
from ``memory effects'' of the original time evolution.  These corrections
can estimated using a bound which for large $t$ and $L$ behaves as 
$\int_{-\infty}^\infty \rmd k\, k^2 \rme^{-t k^2}=O(t^{-3/2})$.  The rest of
the factors can be uniformly bounded using the total energy, resulting in the
bound in (\ref{eq:Tlatticediff}).  We do not know if the bound is optimal,
although this could well be true for generic initial data.  The
worst case scenario for thermalization should be given by initial data in
which all energy is localized to one site.  It would thus be of
interest
to study the solution of (\ref{eq:Tevoleq}) with initial data $q(0)_x=0$
and $p(0)_x=\sqrt{2 L}\, \cf(x{=}0)$ in more detail to settle the issue.

\begin{figure}
\centering
  \myfigure{width=0.7\textwidth}{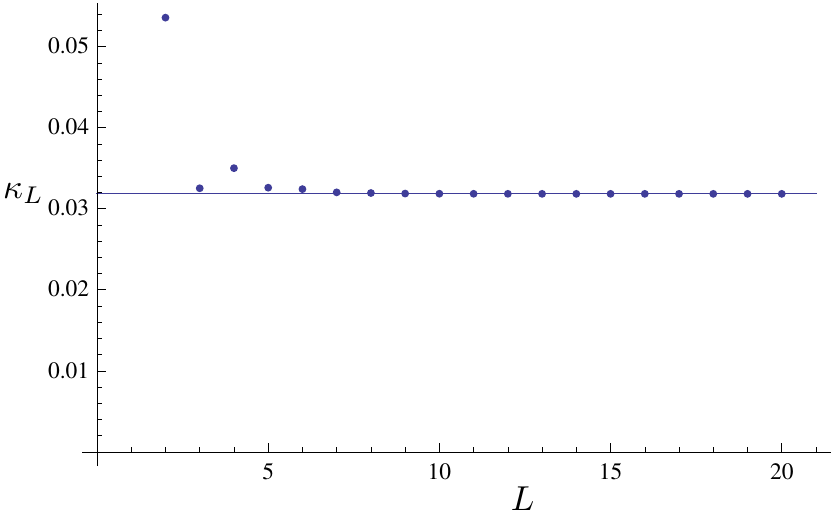}
  \caption{Plot of $\kappa_L$ for nearest neighbor
interactions with $\omega_0=1$ and using $\gamma=6$ for $L=2,3,\ldots,20$.  The
horizontal line depicts the corresponding predicted infinite volume value
$\kappa_\infty$, as explained in the text.\label{fig:kappalim}} 
\end{figure}

The following corollary makes the connection to diffusion more
explicit. Its physical motivation is to show that the Fourier's law can here
be used 
to predict results from temperature measurements, as soon as these are not
sensitive to the 
lattice structure.  Explicitly, 
we assume that the measurement device detects only the cumulative effect of
thermal movement of the particles, say via thermal radiation, and thus can only
measure a smeared temperature profile.   The smearing is assumed to be linear
and given by a convolution with some fixed function $\varphi$ which for
convenience we assume to be smooth and rapidly decaying, i.e., that it should
belong to the Schwartz space.  The Corollary 
implies that then for large systems it is possible to obtain excellent
predictions for future measurements of the temperature profile by first
waiting a time $t_0\gg L^{2/3}$, measuring the temperature profile, and then
using the profile as initial data for the time-dependent Fourier's
law.
The diffusion constant $\kappa_L$ of the Fourier's law depends on the harmonic
dynamics and is given by (\ref{eq:defkappa}) below.  We prove later, in
Corollary \ref{th:ptxcoroll} and Proposition \ref{th:Greenprop},
that the constant remains uniformly bounded away from $0$ and infinity.  Hence,
the present assumptions are sufficient to guarantee normal heat conduction.   
We have also computed the values of $\kappa_L$ numerically for a nearest
neighborhood interaction and plotted these in Fig.~\ref{fig:kappalim}.  The
results indicate that the limit $L\to\infty$ exists and
agrees with $\kappa_\infty =
\gamma^{-1}/(2+\omega_0^2+\omega_0\sqrt{\omega_0^2+4})$ which is the value
obtained in previous works on this model \cite{bkll11,bkll11b,Simon12}.

\begin{corollary}\label{th:maincoroll}
Suppose the assumptions of Theorem \ref{th:maintprof} hold, $L\ge L_0$, and
all second moments of the initial field exist. 
Let $T_{t,x}$ denote the corresponding solution to (\ref{eq:Tevoleq}) and 
$\mathcal{E}:= |\Lambda_L|^{-1}\mean{H_L}$ the energy density.
For any kernel function $\varphi\in \mathcal{S}(\R)$ and initialization
time $t_0>0$, define the corresponding observed temperature profile by
\begin{align}
  \Tobs (t,\xi) := \sum_{y\in \Z} \varphi(\xi-y) T_{t_0+t,y \bmod \Lambda_L}\,
,\quad t\ge 0,\ \xi\in L\T.
\end{align}
Let the predicted temperature profile be defined as the solution of the
diffusion equation on the circle $L\T$ with initial data $\Tobs (0,\cdot)$,
i.e., let $\Tpred\in C^{(2)}([0,\infty)\times L\T)$ be the unique solution to
the Cauchy problem
\begin{align}\label{eq:heateq}
 \partial _t \Tpred(t,\xi) = \kappa_L \, \partial_\xi^2 \Tpred(t,\xi)\, ,\quad
\Tpred(0,\xi)=\Tobs (0,\xi)
\end{align}
where $t\ge 0$ and $\xi\in L\T$, and the diffusion constant is defined by
\begin{align}\label{eq:defkappa}
  \kappa_L:= \sum_{y\in \Lambda_L} y^2 \tilde{p}_{y} >0\, .
\end{align}

Then there is a constant $C>0$, independent of the initial data and of $L$,
$\varphi$ and $t_0$, such that
\begin{align}\label{eq:Taccuracy}
& \left|\Tpred(t,\xi)-\Tobs (t,\xi)\right|\le C \mathcal{E} L t_0^{-3/2}
\Bigl(\norm{\varphi}_1 + \sup_\xi \sum_{y\in \Z} |\varphi(\xi-y)| \Bigr)
+ C \mathcal{E} \sum_{|n| \ge L/2} \left|\FT{\varphi}(n/L)\right| \, ,
\end{align}
for all $t\ge 0$ and $\xi\in L\T$.  In particular, if $\varphi$ is a
``macroscopic averaging kernel'' and satisfies additionally $\varphi\ge 0$,
$\int_\R\!\rmd x\, \varphi(x)=1$ and $\FT{\varphi}(p)=0$ for all $|p|\ge
\frac{1}{2}$, then for the same $C$ as above
\begin{align}\label{eq:Taccuracy2}
& \left|\Tpred(t,\xi)-\Tobs (t,\xi)\right|\le 2 C \mathcal{E} L t_0^{-3/2}\, .
\end{align}
\end{corollary}

The rest of this section is used for proving the above statements.  However,
as the 
arguments get somewhat technical and will not be used in the remaining
sections, it is possible to skip over the details in the first reading. We begin
with a Lemma collecting the main consequences of our assumptions.
\begin{lemma}\label{th:Amainlemma}
  Suppose $\Phi$ and $\gamma$ satisfy the conditions in Assumption
\ref{th:gammaassump}.  
  Use (\ref{eq:Aiexp}) to define $\FT{A}^i:[0,\infty)\times \T\to \R$ for
$i=1,2$, and set
$\delta_0 = \omega_0^2/\gamma$.  Then we can find constants
$c_0,c_1,c_2,\gamma_2>0$, $t_1\ge 0$, such that
  \begin{enumerate}
    \item The functions $\FT{A}^i$, $i=1,2$, belong to $C^{(1)}([0,\infty)\times
\T)$. 
    \item \label{it:Abounds} $|\FT{A}_t^i(k)|,|\partial_t \FT{A}_t^i(k)|\le
c_0\rme^{-\delta_0 t}$ and $|\partial_k \FT{A}_t^i(k)|\le c_0\rme^{-\delta_0
t/2}$ for every $i=1,2$, $t\ge 0$ and $k\in \T$.
    \item\label{it:Alowerb} $-\FT{A}_t^2(k)\ge c_1\rme^{-\gamma t/2}$ for all
$t\ge t_1$ and $k\in \T$.
    \item\label{it:Atildeb} The functions $\tilde{A}^i_{t,x}:= \int_{\T}\! \rmd
k\, \rme^{\ci 2\pi x\cdot k} \FT{A}_t^i(k)$, $x\in \Z$,
    satisfy $|\tilde{A}^i_{t,x}|\le c_2 \rme^{-\delta_0 t/2 - \gamma_2 |x|}$ for
all $i=1,2$, $t\ge 0$, $x\in \Z$.
  \end{enumerate}
\end{lemma}
\begin{proof}
The first item follows straightforwardly from the definitions.  As in the
statement, set  $\delta_0 := \omega_0^2/\gamma>0$ and recall the functions $u$
and $\mu_\pm$ defined in (\ref{eq:defmupm}).
Since $\omega(k)^2/\mu_+(k)= \mu_-(k)< \mu_+(k)\le \gamma$, we have
$\mu_+(k)>\mu_-(k)\ge \delta_0$ for all $k$ and thus
a direct computation shows that $c_0$ for the first two   upper bounds in item \ref
{it:Abounds} can be found.  The bound for $|\partial_k \FT{A}_t^i(k)|$ follows
similarly, using the estimate $t \rme^{-t\delta_0/2}\le \rme^{-1} 2/\delta_0$
and possibly increasing $c_0$ to accommodate the extra factors resulting from
taking the derivative, such as $\max_k |\omega'(k)|<\infty$.  The lower bound in
item \ref{it:Alowerb} is a direct consequence of the identity $-\FT{A}_t^2(k) =
(2 u)^{-1} \mu_- \rme^{- t \mu_-} \left(1-\rme^{-t 2 u} \mu_+/\mu_-\right)$
where
$u=u(k)\ge u_0:=\sqrt{(\gamma/2)^2-\max_k \FT{\Phi}(k)}>0$ and $\mu_-\le
\gamma/2$.  (We may define, for instance, $t_1:=(2 u_0)^{-1} \ln (2
\gamma^2/\omega_0^2)$ and $c_1:=\omega_0^2/(2\gamma^2)$.)

All of the maps $k\mapsto \FT{A}_t^i(k)$ can be represented as a composition of
a function analytic on $\C\setminus\{0\}$ and the function $u$ (note that
$\omega(k)^2=(\gamma/2)^2-u(k)^2$).  Then, 
by assumption,  $0<u_0\le u(k)\le \sqrt{(\gamma/2)^2-\omega_0^2}\le
\frac{1}{2}\gamma- \delta_0$ for real $k$,  and there is a strip on which
$(\gamma/2)^2-\omega(z)^2$ is an analytic, $1$-periodic function.  Therefore, we
can find $\vep_0>0$ such that $u(z)$ is an analytic, $1$-periodic continuation
of $u$ to  a neighborhood of $U:=\R+\ci [-\vep_0,\vep_0]$ with
$u_0/2\le \re u(z)\le \frac{1}{2}\gamma- \delta_0/2$.  Hence, $\FT{A}_t^i|_{u\to
u(z)}$ is also analytic and $1$-periodic on $U$ and it is bounded there by $4
u_0^{-1} (1+(\gamma/2)^2+\max_{z\in U} |u(z)|^2)\rme^{-\delta_0 t/2}$ for
$i=1,2$, $t\ge 0$.  
Therefore, Cauchy's theorem can be used to change the integration contour in the
definition of $\tilde{A}^i_{t,x}$ from $[-1/2,1/2]$ to $[-1/2,1/2]+\ci \sign(x)
\vep_0$ without altering the value of the integral.  Thus we can define
$\gamma_2 := 2 \pi\vep_0$ and find $c_2>0$ independent of $x,t,i$ such that
$|\tilde{A}^i_{t,x}|\le  c_2 \rme^{-\delta_0 t/2 - \gamma_2 |x|}$ for all
$i=1,2$, $t\ge 0$, $x\in \Z$.  This concludes the proof of the Lemma.
\end{proof}

\begin{corollary}\label{th:ptxcoroll}
  Suppose $\Phi$ and $\gamma$ satisfy the conditions in Assumption
\ref{th:gammaassump} and let $\delta_0,\gamma_2,t_1$
  be constants for which Lemma \ref{th:Amainlemma} holds.  For each $L\ge 1$
define
  \begin{align}
    p_{t,x} := 2 \gamma \left(\int_{\Lambda_L^*} \!\rmd k\, \rme^{-\ci 2\pi
k\cdot x} \FT{A}^2_t(k)\right)^2\, ,\quad
    \rho_t := \sum_{y\in\Lambda_L} p_{t,y}\, ,
  \end{align}
  for $t\ge 0$ and $x\in \Lambda_L$.
 Then there are constants $C_0,C_1,C_2>0$, all independent of $L$, such that
  \begin{enumerate}
    \item $t\mapsto p_{t,x}$ belongs to $C^{(1)}([0,\infty))$ for all $x\in
\Lambda_L$.
    \item\label{it:psymm} $p_{t,-x} = p_{t,x}$ for all $x,t$.
    \item\label{it:ptxbounds} $p_{t,x}\ge 0$ and $|\partial_t p_{t,x}|,
p_{t,x}\le C_0 \rme^{-\gamma_2 |x|-\delta_0 t}$ for all $x\in \Lambda_L$ and
$t\ge 0$.
    \item\label{it:ptxnormalization}  $\int_0^\infty \!\rmd t \sum_{x\in
\Lambda_L} p_{t,x} =1$.
    \item\label{it:bint}  $\int_0^\infty \!\rmd t \sum_{x\in \Lambda_L} t
p_{t,x} =\gamma^{-1}$.    
\item\label{it:rhotbounds} $0\le \rho_t\le C_1 \rme^{-2 \delta_0 t}$, for all
$t$, and $\rho_t\ge C_2 \rme^{-\gamma t}$ for $t\ge t_1$.
  \end{enumerate}
\end{corollary}
\begin{proof}
The first item follows directly from the corresponding item in Lemma
\ref{th:Amainlemma}.  Consider then a fixed $L\ge 1$
and denote $A^2_{t,x}:=\int_{\Lambda_L^*} \!\rmd k\, \rme^{\ci 2\pi k\cdot x}
\FT{A}^2_t(k)$ for $x\in \Lambda_L$, $t\ge 0$.  
Since $\FT{A}^2_t(k)\in \R$ with $\FT{A}^2_t(-k)=\FT{A}^2_t(k)$, we have
$A^2_{t,x}\in \R$ with $A^2_{t,-x} = A^2_{t,x}$.
Hence, $p_{t,-x} = p_{t,x}$ and $p_{t,x}\ge 0$.  In particular, item
\ref{it:psymm} holds.

By Lemma \ref{th:Amainlemma}, the Fourier-transform of $\FT{A}^2_t(k)$, denoted
$\tilde{A}^2_{t,x}$, is absolutely summable, and 
thus $\FT{A}^2_t(k) = \sum_{y\in \Z}\rme^{-\ci 2\pi y\cdot k} \tilde{A}^2_{t,y}$
for every $k\in \T$.  Inserting the formula in the definition of $A^2_{t,x}$
yields
\begin{align}\label{eq:defAtx}
& A^2_{t,x}=
\sum_{n\in \Z} \tilde{A}^2_{t,x+n L}\, ,
\end{align}
for all $x,t$.
In the above sum, the definition of $\Lambda_L$ implies that $|x+n L|\ge L(
|n|-1) + L/2\ge  L( |n|-1) + |x|$ if $n\ne 0$.
The exponential bound in item \ref{it:Atildeb} of Lemma \ref{th:Amainlemma} thus
shows that
$|A^2_{t,x}| \le c_2 \rme^{-\delta_0 t/2 - \gamma_2 |x|} (1+2/(1-\rme^{-\gamma_2
L}))$.  As mentioned above, $A^2_{t,x}\in \R$ and thus $0\le p_{t,x}\le 2 \gamma
c_0 c_2 (1+2/(1-\rme^{-\gamma_2}))  \rme^{-\delta_0 t 3/2 - \gamma_2 |x|}$. 
Since $\partial_t  p_{t,x} =
4\gamma A^2_{t,x}\int_{\Lambda_L^*} \!\rmd k\, \rme^{\ci 2\pi k\cdot
x}\partial_t  \FT{A}^2_t(k)$ we find using item \ref{it:Abounds} in Lemma
\ref{th:Amainlemma} that $|\partial_t  p_{t,x}|\le 4 \gamma c_0 c_2
(1+2/(1-\rme^{-\gamma_2}))\rme^{-\delta_0 t 3/2 - \gamma_2 |x|}$.  Choosing
$C_0:=4 \gamma c_0 c_2 (1+2/(1-\rme^{-\gamma_2}))$ thus implies that item
\ref{it:ptxbounds} holds.

Using the discrete Parseval's theorem in the definition of $\rho_t$ yields an
alternative representation $\rho_t = 2 \gamma \int_{\Lambda_L^*} \!\rmd k\,
(\FT{A}^2_t(k))^2$.  Hence, the bounds in Lemma \ref{th:Amainlemma} imply that
the bounds in item \ref{it:rhotbounds} hold for the choices $C_1:= 2 \gamma
c_0^2$ and $C_2:=2 \gamma c_1^2$.   A direct computation using the definition of $\FT{A}^2_t(k)$
shows that $\int_0^\infty\!\rmd t\, (\FT{A}^2_t(k))^2=(2\gamma)^{-1}$ independently
of $k$.  Therefore,
$\int_0^\infty \!\rmd t \sum_{x} p_t(x)=\int_0^\infty \!\rmd t \rho_t=1$ and
item \ref{it:ptxnormalization} holds.  The equality $\int_0^\infty
\!\rmd t \sum_{x\in \Lambda_L} t p_{t,x} =\gamma^{-1}$ can be checked
analogously.
This concludes the proof of the Corollary. 
\end{proof}

The normalization in item  \ref{it:ptxnormalization} is a crucial identity which
makes the structure of (\ref{eq:Tevoleq}) to be that of a renewal equation. 
Instead of the explicit computation referred to in the above proof, the identity
can also be inferred by noting that the integral is equal to a $(2,2)$-diagonal
component of
$\int_0^\infty \!\rmd t\, \rme^{-t M_\gamma^T} \begin{pmatrix}
                                               0 & 0\\ 0 & 1
                                             \end{pmatrix}
 \rme^{-t M_\gamma}$.  By the results proven in \cite{bll02} the integral yields
$\mathcal{G}_L^{-1}$, and thus its $(2,2)$-component has only ones on the
diagonal.   

The main properties of the ``source term'', $g_{t,x}$, are summarized in the
following Proposition.
\begin{proposition}\label{th:gtxprop}
  Suppose $\Phi$ and $\gamma$ satisfy the conditions in Assumption
\ref{th:gammaassump} and let $\delta_0=\omega_0^2/\gamma $ as in Lemma
\ref{th:Amainlemma}.  If $L\ge 1$ and all second moments of the initial field
$X(0)$ exist, we define
for $t\ge 0$ and $x\in\Lambda_L$
\begin{align}
g_{t,x} := \sum_{i',i=1,2}\sum_{y',y\in \Lambda_L}  A^i_{t,y}  A^{i'}_{t,y'}
\mean{X(0)^{i'}_{x+y'}  X(0)^{i}_{x+y}}  \, ,
\end{align}
where $A^2_{t,x}:=\int_{\Lambda_L^*} \!\rmd k\, \rme^{\ci 2\pi k\cdot x}
\FT{A}^i_t(k)$.
Then there is a constant $C'$, independent of $L$ and the choice of initial
state,  such that all of the following statements hold with $E_L :=
\mean{H_L(X(0))}<\infty$: 
\begin{enumerate}
  \item $g_{t,x}\ge 0$ for all $t,x$.
  \item\label{it:gtxbound} $\sum_x g_{t,x}\le C'\rme^{-\delta_0 t} E_L$ for all
$t$.
\item\label{it:gtxnorm} $\int_0^\infty\!\rmd t\, \sum_{x\in\Lambda_L} g_{t,x} =
\gamma^{-1} E_L$.
\end{enumerate}
\end{proposition}
\begin{proof}
The assumptions imply that Lemma \ref{th:Amainlemma} holds.  In the following,
the constants $c_0,c_1,c_2,\gamma_2,t_1$ refer to those appearing in the Lemma.

Since it follows from the definition that $g_{t,x} = \mean{(\sum_{i,y} A^i_{t,y}
X(0)^{i}_{x+y} )^2}$, obviously $g_{t,x}\ge 0$.  As in the proof of Corollary
\ref{th:ptxcoroll}, Lemma \ref{th:Amainlemma} implies that now $| A^i_{t,y}|\le
c_2 \rme^{-\delta_0 t /2-\gamma_2|y|} (1+2/(1-\rme^{-\gamma_2}))$.
By the Schwarz inequality $\sum_{i',i,x}  \mean{|X(0)^{i'}_{x+y'} 
X(0)^{i}_{x+y}|}\le 2\mean{\norm{X(0)}^2_2}$.  Since $\norm{X}^2_2 \le H_L(X)
2/\min(1,\omega^2_0)$, we find that item \ref{it:gtxbound} holds for $C'=4 c_2^2
\max(1,\omega_0^{-2}) (1+2/(1{-}\rme^{-\gamma_2}))^4 $.

For the final item, we return to the matrix formulation of $g$.  Since
$A^i_{t,y-x}= (\rme^{-t M_\gamma})^{i2}_{yx}$ by periodicity, we have $g_{t,x} =
\mean{((\rme^{-t M^T_\gamma}X(0))^{2}_{x} )^2}$ and thus 
\begin{align}
  \sum_{x\in\Lambda_L} g_{t,x}=\mean{X(0)^T\rme^{-t M_\gamma}P^{(2)}\rme^{-t
M^T_\gamma}X(0)}\, ,\quad \text{where}\quad P^{(2)}:= \begin{pmatrix}
                                               0 & 0\\ 0 & 1
                                             \end{pmatrix}\, .
\end{align}
Therefore, $\int_0^\infty \!\rmd t\, \sum_{x\in\Lambda_L} g_{t,x} = \mean{X(0)^T
Q X(0)}$ where $Q:=\int_0^\infty \!\rmd t\,\rme^{-t M_\gamma}P^{(2)}\rme^{-t
M^T_\gamma} $ is the unique real symmetric matrix satisfying $M_\gamma
Q+QM_\gamma^T=P^{(2)}$.  Using the definition of $M_\gamma$ in (\ref{eq:defMM}),
we can then verify that $Q=\frac{1}{2\gamma}\begin{pmatrix}
                                               \Phi_L & 0\\ 0 & 1
                                             \end{pmatrix}$.  Thus $X^T Q
X=H_L(X)/\gamma$, and we can conclude that item \ref{it:gtxnorm} holds.
\end{proof}

Renewal equations can conveniently be studied via Laplace transforms.  We have
included a proof in Appendix \ref{sec:diffkernel} how the above bounds allow
for an 
explicit representation of the solution of (\ref{eq:Tevoleq}) for any continuous
``initial data'' $g_{t,x}$.
The solution is also unique, at least in the class of continuous functions. We
denote the solution by $T_{t,x}$ and conclude that  
\begin{align}
& T_{t,x} =  g_{t,x} + \int_0^t \!\rmd s\, \sum_{y\in \Lambda_L} G(t-s,x-y)
g_{s,y}
\, ,
\end{align}
where for any $\vep>0$
\begin{align}\label{eq:defAkernel}
    G(t,x):= p_{t,x} + \int_{\Lambda_L^*} \!\rmd k\, \rme^{\ci 2\pi k\cdot x}
\int_{\vep-\ci \infty}^{\vep+\ci \infty}\! \frac{\rmd \lambda}{2\pi \ci} \,
\rme^{\lambda t} \frac{\FT{p}(\lambda,k)^2}{1-\FT{p}(\lambda,k)} \, ,
  \end{align}
and  $\FT{p}(\lambda,k) := \int_0^\infty\! \rmd s \, \sum_{y\in \Lambda_L}
p_{s,y} \rme^{-s\lambda}\rme^{-\ci 2 \pi k \cdot y}$ is analytic for $\re
\lambda>-\delta_0$.  The estimates proven in Proposition \ref{th:renewalprop}
also imply that $|\FT{p}(\lambda,k)|\le C'/(1+|\lambda|)$ for all $k$ if $\re
\lambda\ge -\delta_0/2$, where $C'$ can be chosen independently of $L$.  In
particular, the above integral is absolutely convergent for any choice of
$\vep>0$.

Many of the properties below could be derived more easily by relying on standard
results, such as the implicit function theorem.  However, for such bounds to be
useful here, it is crucial to obtain them with $L$-independent constants.  
To convince the reader that no
$L$-dependence is sneaking in, we provide here detailed estimates with
examples of such 
$L$-independent constants albeit  at the cost of some repetition of standard 
computations.  No claim is made that the given choices for the constants
would be optimal.

\begin{proposition}\label{th:Greenprop}
Suppose $\Phi$ and $\gamma$ satisfy the conditions in Assumption
\ref{th:gammaassump}, and set $\delta_0 = \omega_0^2/\gamma$.
Then we can find constants $c'_0,\delta,\vep_0,\beta>0$  and $L_0\ge 1$, such
that $\beta\le 
\delta_0/2$ and
for all $L\ge L_0$, $x\in\Lambda_L$, $t\ge
0$,
\begin{align}
  G(t,x) = \int_{\Lambda_L^*} \!\rmd k\, \rme^{\ci 2\pi k\cdot x} a(k) \rme^{- t
R(k)} + \Delta(t,x)\, ,
\end{align}
where 
\begin{align}
  |\Delta(t,x)|\le c'_0 \rme^{-\delta t}\, .
\end{align}
Here $a(k)$ and $R(k)$ are defined  for $|k|> \vep_0$ by $a(k)=0$ and $R(k)=\beta$, and for $|k|\le \vep_0$
there is a unique $R(k)\in[0, \beta]$, such that 
\begin{align}\label{eq:defRk}
  1 = \int_0^\infty\! \rmd s \, \rme^{s R(k)} \sum_{y\in \Lambda_L} p_{s,y}
\cos(2 \pi k \cdot y)\, ,
\end{align}
and then
\begin{align}\label{eq:defak}
  \frac{1}{a(k)} = \int_0^\infty\! \rmd s \, s \rme^{s R(k)} \sum_{y\in
\Lambda_L} p_{s,y} \cos(2 \pi k \cdot y) \, .
\end{align}

In addition, we can choose the constants so that there are
$c'_1,c'_2,c'_3,c'_4,\kappa'>0$, all independent of $L$, such that 
for all $L\ge L_0$ and $k\in \Lambda_L^*$ with $|k|\le \vep_0$ all of the
following estimates hold:
\begin{enumerate}
  \item $0<a(k)\le c'_0$.
  \item\label{it:kappapbound} $\int_0^\infty\! \rmd s \, \sum_{y\in \Lambda_L} y^2 p_{s,y}\ge
\kappa'$.
  \item $c'_1 k^2 \le R(k)\le c'_2 k^2$ and 
$ |\FT{D}(k;L)-R(k)|\le c'_4 k^4$
with 
  \begin{align}\label{eq:defDk}
    \FT{D}(k';L):= \gamma \sum_{y\in \Lambda_L} (1-\cos(2\pi k'\cdot y))
\int_0^\infty \!\rmd s \, p_{s,y} \, ,\quad  k' \in \Lambda_L^*\, .
  \end{align}
  In addition, we may assume $\FT{D}(k';L)\ge c'_3 \min (|k'|,\vep_0)^2$ for
  all $k' \in 
\Lambda_L^*$.
\end{enumerate}
\end{proposition}
\begin{proof}  
The goal is to use Cauchy's theorem to move the integration contour in
(\ref{eq:defAkernel}) to the left half-plane, in which case the factor
$\rme^{\lambda t}$ produces exponential decay in time.  To do this, it is
crucial to study the zeroes of $1-\FT{p}$ since these will correspond to poles
of the integrand determining the dominant modes of decay.
For notational simplicity, let us for the moment consider some fixed $k_0\in
\Lambda_L^*$ and set $F(\lambda) := 1-\FT{p}(\lambda,k_0)$.  As proven in
Appendix \ref{sec:diffkernel}, $|\FT{p}(\lambda,k_0)|<1$ if $\re \lambda>0$ and
thus $F$ is an analytic function for $\re \lambda>-\delta_0$ which has no zeroes
in the right half plane.  It turns out that under the present assumptions, in
particular, when the nondegeneracy condition in Assumption
\ref{th:gammaassump} holds, only the case with small 
$k_0$ and $\lambda$ will be relevant, and we begin by considering that case.

Suppose first that $\lambda,\lambda_0\in \C$ with $\re \lambda, \re\lambda_0>
-\delta_0$, and $n\in \N$, with $n=0$ also allowed.  The derivatives of $\FT{p}$
can be computed by differentiating the defining integrand. Therefore, the
$n$:th derivative of $F$ is equal to $\cf(n=0)+(-1)^{n+1} \int_0^\infty\! \rmd s
\, \sum_{y\in \Lambda_L} p_{s,y} s^n \rme^{-s\lambda}\rme^{-\ci 2 \pi k_0 \cdot
y} $.
Thus by item \ref{it:rhotbounds} in Corollary \ref{th:ptxcoroll}, for any $n \ge
0$, we have 
\begin{align}\label{eq:Fshift1}
 & |F^{(n)}(\lambda)- F^{(n)}(\lambda_0)|
  \le \int_0^\infty\! \rmd s \, \rho_s s^{n} |\rme^{-s \lambda_0}-\rme^{-s
\lambda}|
   \nonumber \\ & \quad
   \le  |\lambda-\lambda_0| C_1 \int_0^\infty\! \rmd s \, s^{n+1}
\rme^{-\delta_0 s} = |\lambda-\lambda_0| n! C_1\delta_0^{-(n+2)} \, .
\end{align}
Here, the second bound can be derived for instance from the representation
$\rme^{-s \lambda_0}-\rme^{-s \lambda} = 
\int_0^1\!\rmd r\, s (\lambda-\lambda_0)
\rme^{-s(\lambda+r(\lambda_0-\lambda))}$ where in the exponent for any $r$ the
real part is bounded by $\delta_0 s$.  Since $p_{t,-y}=p_{t,y}$, we also have
\begin{align}
  & F'(0) = \int_0^\infty\! \rmd s \, \sum_{y\in \Lambda_L} s p_{s,y} \cos(2 \pi
k_0 \cdot y)
  = \int_0^\infty\! \rmd s \, s \rho_s -
  \int_0^\infty\! \rmd s \, \sum_{y\in \Lambda_L} s p_{s,y} 2 \sin^2(\pi k_0
\cdot y)
   \nonumber \\ & \quad
  \ge  C_2 \int_{t_1}^\infty\! \rmd s \, s \rme^{-\gamma s} -
  2 \pi^2 |k_0|^2 \int_0^\infty\! \rmd s \, \sum_{y\in \Lambda_L} y^2 s p_{s,y}
     \nonumber \\ & \quad
  \ge  C_2 (1+\gamma t_1) \gamma^{-2} \rme^{-\gamma t_1} -
  4 \pi^2 |k_0|^2 C_0 \delta_0^{-2} \sum_{n=1}^\infty n^2 \rme^{-\gamma_2 n}\, ,
\end{align}
where we have used that $|\sin x|\le |x|$, for any $x\in\R$, and applied the
bounds in Corollary \ref{th:ptxcoroll}.
Here the constants $b_0 :=   C_2 (1+\gamma t_1) \gamma^{-2} \rme^{-\gamma t_1}
$ and 
$C_4 := 4 \pi^2 C_0 \delta_0^{-2} \sum_{n=1}^\infty n^2 \rme^{-\gamma_2 n}$ are
strictly positive and independent of $L$.  Therefore, so is
$\vep_1:=\sqrt{b_0/(2 C_4)}$ and 
we can conclude that whenever $|k_0|\le \vep_1$, we have $F'(0)\ge b_0/2$.  

Consider then the case $|k_0|\le \vep_1$, with $\vep_1>0$ defined above.  Let
$r_0>0$ be given such that $r_0<\delta_0$ and suppose that 
$\lambda$ satisfies $|\lambda|\le r_0$.  Then $F'(0)\ge b_0/2$ and
by (\ref{eq:Fshift1}) we have 
$|F'(\lambda)- F'(0)|  \le r_0 C_1\delta_0^{-3} $.  Hence, $\re F'(\lambda) \ge
b_0/2- r_0 C_1\delta_0^{-3} $.  
We set $r_0:= \min(\delta_0/2,b_0 \delta_0^3/(4 C_1))$ which is $L$-independent
and strictly positive, and conclude that then we have a lower bound $\re
F'(\lambda) \ge b_0/4>0$ for all $|\lambda|\le r_0$.  On the other hand, if
$|\lambda|,|\lambda_0|\le r_0$ with $\lambda\ne \lambda_0$, then 
the identity $F(\lambda) =  F(\lambda_0) + (\lambda-\lambda_0) \int_0^1 \!\rmd
r\, F'(\lambda_0+r(\lambda-\lambda_0))$ implies a bound
\begin{align}\label{eq:implicitbnd}
  \re \frac{F(\lambda)-F(\lambda_0)}{\lambda-\lambda_0} \ge \frac{b_0}{4}>0\, .
\end{align}

Suppose that $\lambda_0$ is a zero of $F$ in the closed ball of radius $r_0$. 
Since $\re F'(\lambda_0)>0$, then $\lambda_0$ has multiplicity one.  Also, by
(\ref{eq:implicitbnd}), we have $|F(\lambda)|\ge |\lambda-\lambda_0| b_0/4$ for
all $|\lambda|\le r_0$, and thus there can then be no other zeros of $F$ in the
ball.  
Since $F(\lambda^*)=F(\lambda)^*$ and $F(r)>0$ for all $r>0$, we can also
conclude that then necessarily $\lambda_0=-R_0$ with $0\le R_0\le r_0$. 
Therefore, if $\lambda= -\beta + \ci \alpha $, with $\beta\ne R_0$, $0\le
\beta\le r_0/2$ and $\alpha$ is real and satisfies $|\alpha|\le r_0/2$,  we may
always use the estimate $|1/F(\lambda)|\le 4 b_0^{-1}/|R_0-\beta|$.

Consider then the case in which there are no zeros of $F$ in the closed ball of
radius $r_0$.  The map $r\mapsto F(r)$ is continuous, it maps real values to
real values, and $F(r_0)>0$.  Hence now $F(-r)>0$ for all $0\le r\le r_0$.  We
apply (\ref{eq:implicitbnd}) with $\lambda_0=-r_0$ to conclude that for all
$|\lambda|\le r_0$ with $\lambda\ne -r_0$
\begin{align}
\re \frac{F(\lambda)}{\lambda+r_0} \ge  \frac{b_0}{4} + \re \frac{F(-r_0)
}{\lambda+r_0}  \ge  \frac{b_0}{4}>0\,.
\end{align}
Therefore, 
\begin{align}
  |F(\lambda)| \ge |\lambda+r_0| \left|\re \frac{F(\lambda)}{\lambda+r_0}\right|
\ge |\lambda+r_0| \frac{b_0}{4}\, .
\end{align}
We can then conclude that $|1/F(\lambda)|\le 8/(r_0 b_0)$ whenever $\lambda=
-\beta + \ci \alpha $ with $0\le \beta\le r_0/2$ and $\alpha$ is real and
satisfies $|\alpha|\le r_0/2$.

The above estimates are sufficient to control the $r_0$-neighborhood of zero for
small $k_0$.  Coming back to general $k_0$ we 
next study the properties of $F$ on the imaginary axis, for $\lambda = \ci
\alpha$ with $\alpha\in \R$.  Using item
\ref{it:ptxnormalization} in Corollary \ref{th:ptxcoroll} and the notations
introduced in 
the proof of the Corollary shows that
\begin{align}\label{eq:reFia}
  & \re F(\ci \alpha) = \int_0^\infty\! \rmd s \,  \sum_{y\in \Lambda_L} p_{s,y}
\left(1-\cos(s \alpha+2\pi k_0 \cdot y)\right)
  \nonumber \\ & \quad
  = 4 \gamma \int_0^\infty\! \rmd s \,  \sum_{y\in \Lambda_L} \left|A^2_{s,y}
\sin\Bigl(\frac{1}{2}s \alpha+\pi k_0 \cdot y\Bigr)\right|^2
    \nonumber \\ & \quad
  = 4 \gamma \int_0^\infty\! \rmd s \,  \int_{\Lambda_L^*}\! \rmd k \, \left|
  \sum_{y\in \Lambda_L}\rme^{-\ci 2 \pi k \cdot y} A^2_{s,y}
\sin\Bigl(\frac{1}{2}s \alpha+\pi k_0 \cdot y\Bigr)\right|^2\, .
\end{align}
Since $k_0\in \Lambda_L^*$, there is $n_0\in \Lambda_L$ such that $k_0 = n_0/L$.
To derive lower bounds for (\ref{eq:reFia}), it suffices to consider the case
in which $n_0\ge 0$, since then for $n_0<0$ we 
can use the symmetry of cosine and apply the bounds derived for the case where
the signs of 
$\alpha$ and $k_0$ are reversed.  

Consider first the case in which $n_0$ is even.  Then there is $n_1\in
\Lambda_L$ such that $0\le n_1 \le L/4$ and $n_0=2 n_1$.  In this case, $k_0/2
\in \Lambda_L^*$,
the Fourier-transform of $y\mapsto A^2_{s,y}$ equals
$\FT{A}^2_s(k)$ and thus by using 
$\sin x = (\rme^{\ci x}-\rme^{-\ci  x})/(2\ci)$ in  (\ref{eq:reFia}) yields 
\begin{align}\label{eq:reFiabounds}
  & \re F(\ci \alpha) 
  =\gamma \int_0^\infty\! \rmd s \, \int_{\Lambda_L^*}\! \rmd k \, 
  \left|\rme^{\ci s\alpha/2} \FT{A}^2_s(k-k_0/2)-
  \rme^{-\ci s\alpha/2}\FT{A}^2_s(k+k_0/2)\right|^2
     \nonumber \\ & \quad
  \ge \gamma \int_{t_1}^\infty\! \rmd s \, \int_{\Lambda_L^*}\! \rmd k \,
\left(\FT{A}^2_s(k-k_0/2)-\FT{A}^2_s(k+k_0/2)\right)^2 
    \nonumber \\ & \qquad
 +2\gamma \int_{t_1}^\infty\! \rmd s \, (1- \cos(s \alpha)) \int_{\Lambda_L^*}\!
\rmd k \, \FT{A}^2_s(k-k_0/2)\FT{A}^2_s(k+k_0/2) 
 \end{align}
where in the last step we used the fact that $\FT{A}^2_s$ are real.  Applying
the lower bound in item 
\ref{it:Alowerb} of Lemma \ref{th:Amainlemma} thus proves that
\begin{align}
  & \re F(\ci \alpha) \ge 2\gamma \int_{t_1}^\infty\! \rmd s \, (1- \cos(s
\alpha)) c_1^2 \rme^{-\gamma s}\, .
\end{align}
For instance by representing the cosine as a sum of two exponential terms, we
find that  $\int_{t}^\infty\! \rmd s \, (1- \cos(s \alpha)) \rme^{-\gamma s} =
\gamma^{-1} \rme^{-\gamma t} \alpha^2/(\alpha^2+\gamma^2)$ if $|\alpha| t \in
2 \pi \Z$.
Therefore, now
\begin{align}\label{eq:imaxisb}
  & \re F(\ci \alpha) \ge 2 c_1^2  \rme^{-\gamma t_1} \rme^{-2 \pi
|\gamma/\alpha|} \frac{1}{1+|\gamma/\alpha|^2}\, .
\end{align}
For any $r>0$, set $\tilde{C}(r)$ to be equal to the right hand side at
$\alpha=r$.  Then 
$\tilde{C}(r)>0$, it is independent of $L$, and we can conclude that 
$\re F(\ci \alpha) \ge \tilde{C}(r)$ 
whenever $|\alpha|\ge r$ and $Lk_0$ is even.

In the remaining cases $n_0$ is odd and positive.  Then there is $n_1\in
\Lambda_L$ such that $0\le n_1 \le L/4$ and $n_0=2 n_1+1$.
Thus we can apply the above estimate for $\re F(\ci \alpha)$ at $k_0-1/L=2
n_1/L$.  On the other hand,
\begin{align}\label{eq:cosdiff}
\left|\cos(s \alpha+2\pi k_0 \cdot y) - \cos(s \alpha+2 \pi (k_0-1/L)\cdot
y)\right| \le \frac{2 \pi |y|}{L}\, ,
\end{align}
and we can conclude that for odd $n_0$ and every $|\alpha|\ge r>0$
\begin{align}
  \re F(\ci \alpha) \ge \tilde{C}(r) - \int_0^\infty\! \rmd s \,  
\sum_{y\in \Lambda_L}
p_{s,y} \frac{2 \pi |y|}{L}
   \ge \tilde{C}(r) - L^{-1} \frac{4 \pi C_0}{\delta_0} \sum_{n=1}^\infty n
\rme^{-\gamma_2 n}\, ,
\end{align}
where in the second inequality we have applied the bounds in item
\ref{it:ptxbounds} of Corollary \ref{th:ptxcoroll}.  Therefore, to every $r>0$
there is  
$\tilde{L}(r)\in \N_+$ such that 
the final bound is greater than $\tilde{C}(r)/2$ for every 
$L\ge \tilde{L}(r)$.   
Thus we can conclude that, if $r>0$ and $L\ge \tilde{L}(r)$, then 
$\re F(\ci \alpha;k_0,L) \ge \tilde{C}(r)/2$ for all $|\alpha|\ge r$ 
and $k_0\in \Lambda^*_L$.

If $\beta$ satisfies $0\le \beta <\delta_0$, then by (\ref{eq:Fshift1}) we
have $|F(-\beta + \ci \alpha)- F(\ci \alpha)|\le \beta C_1\delta_0^{-2}$ for
all real $\alpha$. 
Therefore, if we set 
$\tilde{\beta}(r):=\tfrac{1}{2}\delta_0 \min(1,\tilde{C}(r)\delta_0/(2 C_1))$
for $r>0$, then  
we have found strictly positive, $L$-independent constants such that for any
$r>0$ and  
$L\ge\tilde{L}(r)$
\begin{align}
  \re F(-\beta + \ci \alpha; k_0,L) \ge \frac{1}{4} \tilde{C}(r) >0\, ,
\end{align}
for all $k_0\in \Lambda_L^*$, $|\alpha|\ge r$, and 
$0\le \beta\le\tilde{\beta}(r)\le \delta_0/2$. 

It is now possible to conclude the estimates for the case when
$|k_0|\le \vep_1$.  Recall the earlier definition of $r_0$ and set
$C_5 := \tilde{C}(r_0/2)$, $L_5 := \tilde{L}(r_0/2)$ and 
$\beta_1 := \tilde{\beta}(r_0/2)$.  Assume that $L\ge L_5$.
We use Cauchy's theorem and, when necessary, the residue
theorem to change the integration contour from $\vep'+\ci \R$, $\vep'>0$, to
$-\beta+\ci \R$ with some $\beta>0$.  If there are no zeroes of $F$ in the
closed ball of 
radius $r_0$, we choose $\beta=\beta_0$ with $\beta_0:=\min(\beta_1,r_0/2)$ and
the above results imply that the integrand in (\ref{eq:defAkernel}) is analytic
for $\re \lambda\ge -\beta_0$ and we have $|1/F|\le \max(4/C_5,8/(r_0 b_0))$
on the 
integration contour.  If there are zeroes in the ball, then the zero is unique
and lies at $-R_0$ with $0\le R_0\le r_0$ and 
$1/F$ has a first order pole at $-R_0$.  If $R_0> \beta_0/2$, we choose
$\beta=\beta_0/4<\beta_1$ when the integrand is analytic to the right of the
final 
contour and hence the pole does not contribute.  Then also $|1/F|\le
\max(4/C_5,16/(\beta_0 b_0))$ on the integration contour.  If $R_0\le
\beta_0/2$, we choose $\beta=\beta_0$.  Then the residue theorem can be used to
evaluate the contribution from the pole, and the remaining integral over
$-\beta+\ci \R$ can be bounded using $|1/F|\le \max(4/C_5,8/(\beta_0 b_0))$.

Assume then that $L\ge L_5$ and $|k_0|\le \vep_1$.
Following the above steps, we find that, if
there is $0\le R_0\le \beta_0/2$ such that
\begin{align}
  1 = \int_0^\infty\! \rmd s \, \rme^{s R_0} \sum_{y\in \Lambda_L} p_{s,y}
\cos(2 \pi k_0 \cdot y)\, ,
\end{align}
then $\FT{p}(-R_0,k_0)=1$ and 
\begin{align}
 \int_{\vep'-\ci \infty}^{\vep'+\ci \infty}\! \frac{\rmd \lambda}{2\pi \ci} \,
\rme^{\lambda t} \frac{\FT{p}(\lambda,k_0)^2}{1-\FT{p}(\lambda,k_0)}
 = \frac{1}{m(k_0)} \rme^{- t R_0} + \Delta
\, .
\end{align}
Here $m(k_0)=F'(-R_0)$, implying
\begin{align}
  m(k_0) := \int_0^\infty\! \rmd s \, s \rme^{s R_0} \sum_{y\in \Lambda_L}
p_{s,y} \cos(2 \pi k_0 \cdot y) \ge \frac{b_0}{4}>0\, ,
\end{align}
and there is a constant $C_6>0$, independent of $L$, such that 
\begin{align}
  |\Delta|\le C_6 \rme^{-\beta_0 t}\, .
\end{align}
(Recall that $|\FT{p}(\lambda,k_0)|\le C'/(1+|\lambda|)$ for 
$\re\lambda\ge -\delta_0/2$.) 
If no such $R_0$ can be found, then $\FT{p}(-r,k_0)<1$ for all 
$r\le \beta_0/2$ and one 
of the remaining cases is realized.  Hence, then 
\begin{align}
  \left|\int_{\vep'-\ci \infty}^{\vep'+\ci \infty}\! 
\frac{\rmd \lambda}{2\pi \ci} \,
\rme^{\lambda t} \frac{\FT{p}(\lambda,k_0)^2}{1-\FT{p}(\lambda,k_0)}
\right|\le C'_6
\rme^{-\beta_0 t/4}\, ,
\end{align}
with some $L$-independent $C'_6>0$.

The following Lemma will be used to study the remaining values of $k_0$.
\begin{lemma}\label{th:largeklemma}
  Suppose  Assumption \ref{th:gammaassump} holds.  Then for every 
$\vep>0$ we can find $C(\vep)>0$,
  $L(\vep)\in \N_+$ and $\beta(\vep)\in(0,\delta_0/2]$ such that, if 
$L\ge L(\vep)$ and $k_0\in \Lambda_L^*$ with $|k_0|\ge \vep$, then
  $F(0;k_0,L)\ge C(\vep)$ and $\re  F(\lambda;k_0,L)\ge C(\vep)/2$ 
for all $\lambda$ with $0\ge \re \lambda\ge -\beta(\vep)$.
\end{lemma}
\begin{proof}
Fix $\vep>0$ and let $C_\vep>0$ denote the corresponding constant in
Assumption \ref{th:gammaassump}.   
As proven above, if $L\ge 1$ and $k_0\in \Lambda_L^*$ is such that $L k_0$ is
a even and nonnegative, then  
\begin{align}
F(0;k_0,L)
 =\gamma \int_0^\infty\! 
\rmd s \, \int_{\Lambda_L^*}\! \rmd k \, f_s(k,k_0)\, ,\quad
  f_s(k,k_0):=\left(\FT{A}^2_s(k-k_0/2)-\FT{A}^2_s(k+k_0/2)\right)^2\, .
\end{align}
If $h\in C^{(1)}(\T)$, then $|\int_\T\!\rmd k\, h(k)-\int_{\Lambda_L^*}\! \rmd
k \, h(k)|\le  
\sum_{n\in\Lambda_L} \norm{h'}_\infty \int_{|k-n/L|\le (2 L)^{-1}}\!\rmd k\, 
|k-n/L|\le \norm{h'}_\infty /(2L)$.  
By Lemma \ref{th:Amainlemma}, we can apply this in the above with
$\norm{h'}_\infty\le 8 c_0^2 \rme^{-\delta_0 s}$.  Therefore, 
\begin{align}
\left|F(0;k_0,L)
 -\gamma \int_0^\infty\! \rmd s \, \int_{\T}\! \rmd k \, f_s(k,k_0)\right|
  \le \frac{4 \gamma c_0^2}{L\delta_0} 
  \, .
\end{align}
If $L k_0$ is odd and nonnegative, we have by (\ref{eq:cosdiff})
\begin{align}
|F(0;k_0,L)-F(0;k_0{-}1/L,L)|\le \frac{4\pi C_0}{L\delta_0}  \sum_{n=1}^\infty n \rme^{-\gamma_2 n}  
\end{align}
and also
$|f_s(k,k_0)-f_s(k,k_0{-}1/L)|\le 4 c_0^2 \rme^{-\delta_0 s}L^{-1}$, by Lemma
\ref{th:Amainlemma}.  Therefore, there is an $L$-independent constant $C>0$
such that $\left|F(0;k_0,L) 
 -\gamma \int_0^\infty\! \rmd s \, \int_{\T}\! \rmd k \, f_s(k,k_0)\right|
  \le C/L$
for all $k_0\in \Lambda_L^*$.  If $|k_0|\ge \vep$, then by assumption $\gamma
\int_0^\infty\! \rmd s \, \int_{\T}\! \rmd k \, f_s(k,k_0)\ge \gamma C_\vep $
and hence $F(0;k_0,L)\ge \gamma C_\vep-C/L$.  Thus by choosing $L'(\vep)$ such
that $ L'(\vep) \ge 2 C/(\gamma C_\vep)$ we have $F(0;k_0,L)\ge \gamma
C_\vep/2$ whenever $L\ge L'(\vep)$ and $|k_0|\ge \vep$.  

Consider then some fixed $L\ge L'(\vep)$ and $|k_0|\ge \vep$.
By the earlier results, $|F(\lambda)-F(0)|\le r C_1 \delta_0^{-2}$ if
$|\lambda|\le r<\delta_0$. 
(The constant $C_1$ here should not be confused with $C_\vep$ at $\vep=1$.)
Thus if  
$r_1:= \min(\delta_0/2,\gamma C_\vep\delta_0^2/(4C_1))>0$, then 
$\re F(\lambda)\ge  \gamma C_\vep/4$ for all $|\lambda|\le r_1$.  On the other
hand, if also $L\ge \tilde{L}(r_1/2)$, then we have $\re F(-\beta+\ci
\alpha)\ge \tilde{C}(r_1/2)/4$ for all $0\le\beta \le \tilde{\beta}(r_1/2)$
and $|\alpha|\ge r_1/2$.  Combining the above estimates yields constants such
that the Lemma holds for all $L\ge L(\vep):=
\max(L'(\vep),\tilde{L}(r_1/2))$. 
\end{proof}

We next apply Lemma \ref{th:largeklemma} with $\vep=\vep_1/2>0$.  Set thus
$L_7:=L(\vep_1/2)$, $C_7:=C(\vep_1/2)$, and $\beta_2:=\beta(\vep_1/2)$.
Assume $L\ge L_7$ and $k_0\in \Lambda_L^*$ with $|k_0|\ge \vep_1/2$.  Then by
the Lemma, 
for any $\lambda= -\beta + \ci \alpha $, with $0\le \beta\le \beta_2$ and
$\alpha\in \R$  
we have $|1/F(\lambda)|\le 2/C_7$.  
Therefore, 
we can change the contour to $-\beta_2+\ci \R$
without encountering any singularities.  This proves that there is an
$L$-independent constant $C_8$ such that for $|k_0|\ge \vep_1$ and $L\ge L_7$
\begin{align}
  \left|\int_{\vep'-\ci \infty}^{\vep'+\ci \infty}\! 
\frac{\rmd \lambda}{2\pi \ci} \,
\rme^{\lambda t} \frac{\FT{p}(\lambda,k_0)^2}{1-\FT{p}(\lambda,k_0)}
\right|\le C_8
\rme^{-\beta_2 t}\, .
\end{align}

Collecting the above estimates together proves that there are
constants $c_0'>0$, $\beta'>0$, $L'\in \N_+$,
such that if $L\ge L'$, then for all $k_0\in \Lambda_L^*$ either 
\begin{align}
  \left|\int_{\vep'-\ci \infty}^{\vep'+\ci \infty}\! 
\frac{\rmd \lambda}{2\pi \ci} \,
\rme^{\lambda t} \frac{\FT{p}(\lambda,k_0)^2}{1-\FT{p}(\lambda,k_0)}
\right|\le c_0'
\rme^{-\beta' t}\, ,
\end{align}
or $|k_0|\le \vep_1$ and there are $R(k_0)$ and $a(k_0):=1/m_0(k_0)$ satisfying
(\ref{eq:defRk}) and (\ref{eq:defak}) such that
\begin{align}
  \left|\int_{\vep'-\ci \infty}^{\vep'+\ci \infty}\! 
\frac{\rmd \lambda}{2\pi \ci} \, \rme^{\lambda t} 
\frac{\FT{p}(\lambda,k_0)^2}{1-\FT{p}(\lambda,k_0)}- a(k_0)
\rme^{-t R(k_0)}\right|\le c_0' \rme^{-\beta' t}\, .
\end{align}

We still need to make sure that all the claimed bounds will hold.   From now on
we assume that $L\ge L'$ so that all of the earlier derived
bounds can be used.  Suppose then that $|k_0|\le \vep_1$.
Since $F(-R)\in \R$, for $0\le R\le r_0$,  (\ref{eq:implicitbnd}) implies that
$F(-R)\le F(0)-R b_0/4$ for these $R$.  Therefore, if $F(-R(k_0))=0$, we have
$0\le R(k_0)\le 4 F(0)/b_0$.  Since 
\begin{align}\label{eq:Fzeroestimates}
&  F(0) = 2 \int_0^\infty\! \rmd s \,  \sum_{y\in \Lambda_L} p_{s,y} \sin^2(\pi
k_0 \cdot y)
\le 2 \pi^2 k_0^2 \int_0^\infty\! \rmd s \,  \sum_{y\in \Lambda_L} p_{s,y} y^2
\le k_0^2 \frac{4 \pi^2 C_0}{\delta_0} \sum_{n=1}^\infty n^2 \rme^{-\gamma_2
n}\, ,
\end{align}
we can conclude that with $c'_2:=(4\pi)^2 C_0/(b_0 \delta_0) \sum_{n=1}^\infty
n^2 \rme^{-\gamma_2 n}$ we have $R(k_0)\le  c'_2 k_0^2$. 
Also, whenever $L \ge L_8:=\max(L',L_7,1/\vep_1)$, there exists 
$k_\vep \in [\vep_1/2,\vep_1]\cap \Lambda_L^*$, and by Lemma
\ref{th:largeklemma}, then 
$F(0;k_\vep,L)\ge C_7>0$.  Therefore, for $L\ge L_8$ also
$\int_0^\infty\! \rmd s \,  \sum_{y\in \Lambda_L} 
p_{s,y} y^2\ge C_7/(2 \pi^2 \vep_1^2)$.  Thus if we set 
$\kappa':=C_7/(2 \pi^2 \vep_1^2)$, then item \ref{it:kappapbound} holds.

To get a lower bound, we assume $L\ge L_8$, and use the fact that $|\sin x|\ge
|x|2/\pi$ for all $|x|\le \pi /2$.  This shows that if $0<\vep_0\le \vep_1$,
then for all $|k_0|\le \vep_0$
\begin{align}
&  F(0) \ge k_0^2 8 \int_0^\infty\! \rmd s \,  \sum_{y\in \Lambda_L} y^2 
p_{s,y} \cf\!\left(|y|\le \frac{1}{2 \vep_0}\right)
\ge k_0^2 8 \Bigl[ \kappa' - \frac{2 C_0}{\delta_0} \sum_{n> 1/(2\vep_0)} n^2
\rme^{-\gamma_2 n}
\Bigr] \, .
\end{align}
Therefore, by choosing any $L$-independent $\vep_0\in \left(0,\sqrt{r_0/c'_2}\right]$ such that
$\vep_0\le \vep_1$ and for which $\sum_{n> 1/(2\vep_0)} n^2 \rme^{-\gamma_2 n} \le \kappa'\delta_0/(4
C_0)$, we have 
$F(0)\ge k_0^2 4 \kappa'$ for all $|k_0|\le \vep_0$.   It follows from
(\ref{eq:Fshift1}) that $F(0)-F(-R)=|F(0)-F(-R)|\le R C_1 \delta_0^{-2}$ for
$0\le R\le r_0$.  
Therefore, if $F(-R(k_0))=0$, we have $R(k_0)\ge \delta_0^2 F(0)/C_1\ge c'_1
k_0^2$ with $c'_1:= 4 \kappa' \delta_0^2/C_1>0$ for all $|k_0|\le \vep_0$. 

Set then $\beta:=c'_2 \vep_0^2 \in (0,r_0]$.
Collecting the above estimates together, we can now conclude that if $|k_0|\le
\vep_0$, then $F(-\beta)\le (c'_2 k_0^2-\beta) b_0/4\le 0$, and thus there is a
unique $R(k_0)\in [0,\beta]$ such that $F(-R(k_0))=0$, and then also $c'_1 k_0^2
\le R(k_0)\le c'_2 k_0^2$.  As $|k_0|\le \vep_1$, then $0<a(k_0)\le
b_0/2$
for $a(k_0):=1/m_0(k_0)$.  If $\vep_0<|k_0|\le \vep_1$, then either $R(k_0)\ge
\delta_0^2 F(0;k_0)/C_1$, or there is no zero of $F$ in $[0,\beta_0/2]$. 
In the first case, we have $F(0;k_0,L)\ge C(\vep_0)>0$ for all $L\ge L(\vep_0)$
and thus then $|\rme^{-t R(k_0)}/m_0(k_0)|\le b_0/2 \rme^{-t
  C(\vep_0)\delta_0^2/C_1}$.   
In the second case, the previous estimates apply.
Thus by setting
$\delta:=\min(C(\vep_0) \delta_0^2/C_1,\beta_2,\beta_0/4)>0$ we have also
proven the 
exponential upper bound for the correction.  (Note that $|p_{t,x}|\le C_0
\rme^{-\delta_0 t}$ decays always faster than $\rme^{-\delta t}$.)

Finally, define $\FT{D}(k')$ by (\ref{eq:defDk}) for all $k' \in \Lambda_L^*$.  
Comparing the definition to (\ref{eq:Fzeroestimates}) shows that then in fact
$\FT{D}(k') = \gamma F(0;k')$.  Therefore, if $L$ is large enough, then by the
above estimates, we have $\FT{D}(k')\ge \gamma C(\vep_0)$ for $|k'|> \vep_0$,
and 
$\FT{D}(k')\ge 4 \gamma \kappa' (k')^2$ for $|k'|\le \vep_0$.  Hence,  we can
arrange that $\FT{D}(k')\ge c'_3 \min (|k'|,\vep_0)^2$ for some $c'_3>0$,
independent of $L$,
as claimed in the Proposition.  Using item \ref{it:bint} in Corollary
\ref{th:ptxcoroll} and the above estimates
then shows that whenever $|k|\le \vep_0$
\begin{align}
 &  \gamma^{-1}\left(R(k)-\FT{D}(k)\right) = -1+\int_0^\infty\! \rmd s \, 
  \sum_{y\in \Lambda_L} p_{s,y} \left(s R(k)+\cos(2 \pi k \cdot y) \right)
\nonumber \\ & \quad
 = \int_0^\infty\! \rmd s \, 
  \sum_{y\in \Lambda_L} p_{s,y} s R(k) \left(1-\cos(2 \pi k \cdot y)\right)
\nonumber \\ & \qquad
  +\int_0^\infty\! \rmd s \, \left(1+s R(k) - \rme^{sR(k)}\right)
   \sum_{y\in \Lambda_L} p_{s,y} \cos(2 \pi k \cdot y)\, ,
\end{align}
where in the second step we used the defining relation of $R(k)$,
equation (\ref{eq:defRk}).  The first term in the sum is bounded by $k^4 c'_2
C_0 (2\pi /\delta_0)^2 \sum_{n=1}^\infty n^2 \rme^{-\gamma_2 n}$, 
and the second one is
bounded by $R(k)^2\int_0^\infty\!\rmd s \rho_s s^2\rme^{s R(k)}\le k^4 2 C_1
(c'_2)^2 \delta_0^{-3}$.  Choosing the sum of the two factors multiplying $k^4$
as $c'_4/\gamma$ then implies $|R(k)-\FT{D}(k)|\le c'_4 k^4$.  This concludes
the proof of  the Proposition.
\end{proof}

The following observation will provide a convenient estimate for the proof of
the main theorem.
\begin{lemma}\label{th:nsqsum}
  $\displaystyle\sum_{n=1}^\infty n^2 \rme^{-\vep n^2}\le 2 \vep^{-\frac{3}{2}}$
for all $\vep>0$.
\end{lemma}
\begin{proof}
Fix $\vep>0$ and consider the function $f(x)=x^2 \rme^{-\vep x^2}$ for $x\ge 0$.
 It is strictly increasing on $[0,x_\vep]$ and strictly decreasing for $x\ge
x_\vep$, with $x_\vep:= \vep^{-\frac{1}{2}}$.
If $\vep<1$, we have $x_\vep>1$, and we set $n_\vep\ge 1$ as the integer part of
$x_\vep$.
We estimate the sum as an integral over a step function containing values of
$f$. 
This shows that $\sum_{n=1}^{n_\vep-1} n^2 \rme^{-\vep n^2} \le \int_0^{x_\vep}
\!\rmd x f(x)$ and 
$\sum_{n=n_\vep+2}^\infty n^2 \rme^{-\vep n^2} \le \int_{x_\vep}^\infty \!\rmd x
f(x)$.  Hence, $\sum_{n=1}^\infty n^2 \rme^{-\vep 
n^2}\le \int_{0}^\infty \!\rmd x\, f(x) + 2 f(x_\vep) \le \vep^{-\frac{3}{2}}
\Bigl(\int_{0}^\infty \!\rmd y\, y^2 \rme^{-y^2}
+ 1\Bigr)$.  The constant is equal to $1+\sqrt{\pi}/4<2$.

If $\vep\ge 1$, we have $x_\vep\le 1$ and thus $f$ is strictly decreasing for
$x\ge 1$.  As above, this implies $\sum_{n=1}^\infty n^2 \rme^{-\vep n^2}\le
\int_{1}^\infty \!\rmd x\, f(x) + \rme^{-\vep} \le \vep^{-\frac{3}{2}}
\Bigl(\int_{0}^\infty \!\rmd y\, y^2 \rme^{-y^2}
+ (3/(2 \rme))^{3/2}\Bigr)<2 \vep^{-\frac{3}{2}}$.  
\end{proof}

\begin{proofof}{Theorem \ref{th:maintprof}}
Suppose now that $L\ge L_0$ which together with the  assumptions of the Theorem
allows using the formulae and constants given in Proposition \ref{th:Greenprop},
with $\delta_0:=\omega_0^2/\gamma$.  
  In particular, let $\FT{D}(k)$ be defined by the formula (\ref{eq:defDk}),
$a(k)$ for $|k|\le \vep_0$ by (\ref{eq:defak}),  and set
 \begin{align}
   \tau_x := \sum_{y\in \Lambda_L}  \int_{\Lambda_L^*} \!\rmd k\, \rme^{\ci
2\pi k\cdot (x-y)} a(k) \int_0^\infty\!\rmd s \,g_{s,y}\, ,
 \end{align}
 and 
\begin{align}
  (D \tau)_x := \sum_{y\in \Lambda_L} \tilde{p}_y  \left(2 \tau_x- \tau_{x+y} -
\tau_{x-y} \right)\, ,
  \quad \tilde{p}_x :=  \frac{\gamma}{2} \int_0^\infty \!\rmd s \, p_{s,x} \, .
 \end{align}
 Clearly, $ \tilde{p}_x\ge 0$, and since $a(-k)=a(k)$, $\tau_x\in \R$.
By the discussion before Proposition \ref{th:Greenprop} we can now conclude that
the there is a unique continuous solution $T_{t,x}$ to (\ref{eq:Tevoleq}).  
It satisfies
  \begin{align}\label{eq:Tsol}
    T_{t,x} = g_{t,x} + \int_0^t\!\rmd s\, \sum_{y\in \Lambda_L} \Delta(t-s,x-y)
g_{s,y} +
    \sum_{y\in \Lambda_L} \int_{\Lambda_L^*} \!\rmd k\, 
    \rme^{\ci 2\pi k\cdot (x-y)} a(k) \int_0^t\!\rmd s\, \rme^{-(t-s) R(k)}
g_{s,y}\, .
  \end{align}
  By Proposition \ref{th:gtxprop}, the first term is bounded by $C'
\rme^{-\delta_0 t} E_L$, where $E_L=L\mathcal{E}$ denotes the average total
energy.  Applying also Proposition \ref{th:Greenprop} to the second term shows
that it is bounded by $
   \int_0^t\!\rmd s\, \sum_{y\in \Lambda_L}  g_{s,y} c'_0 \rme^{-\delta (t-s)}
\le 2 C' c'_0 \delta_0^{-1}\rme^{-\delta t} E_L$, since $\delta \le \delta_0/2$.
 In the third term, we separate the term with $k=0$, for which $R(0)=0$.  
By item \ref{it:bint} in Corollary \ref{th:ptxcoroll}  then $1/a(0) =
\int_0^\infty\!\rmd s\, s \rho_s =\gamma^{-1}$, and thus  the $k=0$ term is
equal to $L^{-1} \gamma \int_0^t\!\rmd s\, \sum_{y}g_{s,y}$.  By Proposition
\ref{th:gtxprop} this differs from $\mathcal{E}$ maximally by $\mathcal{E} C'
\gamma/\delta_0 \rme^{-\delta_0 t}$. Therefore, with $c:=C'\max(1+2c'_0\delta_0^{-1},\gamma\delta_0^{-1})$,
  \begin{align}
 &   |T_{t,x}-\mathcal{E}|\le c \mathcal{E} (1+L) \rme^{-\delta t} + 
    c'_0 C' \mathcal{E} \sum_{1\le |n|\le L \vep_0} \rme^{-t c'_1 n^2 L^{-2}}
\int_0^t\!\rmd s\,  \rme^{-\delta_0 s/2}
    \nonumber \\ & \quad
 \le c \mathcal{E} (1+L) \rme^{-\delta t} + 
    4 c'_0 C' \delta_0^{-1} \mathcal{E} \sum_{n=1}^\infty \rme^{-t c'_1 n^2
L^{-2}} 
 \, .
  \end{align}
Now for any $\vep>0$, $\sum_{n=1}^\infty \rme^{-\vep n^2}\le \rme^{-\vep}
\sum_{n'=0}^\infty \rme^{-2 \vep n'}=\rme^{-\vep}/(1-\rme^{-2 \vep})$.  Then for
$\vep = t c'_1 L^{-2}$ we also have $L^2\le 2 t c'_1/(1-\rme^{-2 \vep})$.
Therefore, we can find a constant $C$ such that (\ref{eq:TdiffersfromE}) holds
for $d:=\min(c'_1,\delta L_0^2/2)>0$.

In order to prove the lattice diffusion equation, we come back to
(\ref{eq:Tsol}).
For the third term we now apply the estimates
  \begin{align}
   & \left| \int_0^t\!\rmd s\, g_{s,y} \rme^{-(t-s) R(k)}- \rme^{-t \FT{D}(k)}
    \int_0^\infty\!\rmd s\, g_{s,y}\right|
 \nonumber \\ & \quad
 \le \left| \rme^{-t R(k)}- \rme^{-t \FT{D}(k)}\right| \int_0^\infty\!\rmd s\,
g_{s,y}
 + \int_t^\infty\!\rmd s\, g_{s,y}
 + \int_0^t\!\rmd s\, g_{s,y} \rme^{-(t-s) R(k)} \left|1 - \rme^{-s R(k)}\right|
\, .
  \end{align}
Splitting the final integral into two parts at $s=t/2$ then yields 
  \begin{align}
   & \sum_{y\in \Lambda_L}  \left| \int_0^t\!\rmd s\, g_{s,y} \rme^{-(t-s)
R(k)}- \rme^{-t \FT{D}(k)}
    \int_0^\infty\!\rmd s\, g_{s,y}\right|
 \nonumber \\ & \quad
 \le C'\delta_0^{-1} E_L \left[  t \left|\FT{D}(k)-R(k)\right| \rme^{-t
\min(R(k),\FT{D}(k))} + \rme^{-\delta_0 t}
 + \delta_0^{-1}  R(k) \rme^{-\frac{t}{2} R(k)} + \delta_0 t\rme^{-\frac{t}{2}
\delta_0} \right] \, .
  \end{align}
  Using the known properties of $R$ and $\FT{D}$, we can now conclude that there
are constants $c,m>0$, independent of $L$ and the initial state, such that 
  \begin{align}\label{eq:Tdiffest}
  &  \left|  T_{t,x} -
    \int_{\Lambda_L^*} \!\rmd k\, \rme^{-t \FT{D}(k)} 
   \sum_{y\in \Lambda_L} \rme^{\ci 2\pi k\cdot (x-y)} a(k) \int_0^\infty\!\rmd
s\, g_{s,y} \right|
  \nonumber \\ & \quad
 \le c  E_L \left[ \rme^{-\frac{\delta}{4} t} + \int_{\Lambda_L^*} \!\rmd k\,
\cf(|k|\le \vep_0)  k^2 \rme^{-t m k^2}
 \right]
   \, .
  \end{align}
To arrive at the above bound, we choose $m:=\min(c'_3,c'_1)/2$ and
estimate $t k^2 \rme^{-2 m t k^2}\le m^{-1} \rme^{-m t k^2}$.  Here, by Lemma
\ref{th:nsqsum}, 
$\int_{\Lambda_L^*} \!\rmd k\, \cf(|k|\le \vep_0)  k^2 \rme^{-t m k^2}\le L^{-3}
2 \sum_{n=1}^\infty n^2 \rme^{-t m L^{-2} n^2}
\le 4 m^{-3/2} t^{-3/2}$.  Thus for $C:=4 (\delta^{-3/2}+ m^{-3/2}) C'$, the right
hand side of (\ref{eq:Tdiffest}) is bounded by $C E_L t^{-3/2}$.  On the other
hand, since the Fourier-transform of the operator $D$ is equal to multiplication
by $\FT{D}(k)$, we can now conclude that (\ref{eq:Tlatticediff}) holds.
\end{proofof}

\begin{proofof}{Corollary \ref{th:maincoroll}}
  Fix an allowed $L\ge L_0$, and some $t_0>0$ and the function $\varphi$.  For
any initial data $f_0\in C^{(2)}(L\T)$
  the solution of the heat equation (\ref{eq:heateq}) on the circle is standard
and can be done using Fourier series.  Explicitly, then
   \begin{align}\label{eq:diffsolFT}
f(t,\xi) = \sum_{n\in \Z} \rme^{\ci 2 \pi n\cdot \xi/L} \rme^{-t \kappa_L (2 \pi
n/L)^2} \FT{f}_0(n)
\end{align}
where  
  \begin{align}
     \FT{f}_0(n) := \frac{1}{L}\int_{L\T}\!\rmd \xi\, \rme^{-\ci 2 \pi n\cdot
\xi/L} f_0(\xi) \, .
  \end{align}
On the other hand, a solution to the periodic heat equation coincides with the
solution to the heat equation on $\R$ with periodic initial data, and thus also
$|f(t,x)|\le \norm{f}_\infty$ for all $t,x$.

As intermediate approximations, set $\tau_{t,x}:=(\rme^{-t D} \tau)_x$, for
$t\ge 0$, $x\in\Lambda_L$, and, for $t\ge 0$, $\xi\in L\T$,  set
$\tilde{T}(t,\xi) := \sum_{y\in\Z} \varphi(\xi-y) \tau_{t+t_0,y\bmod
\Lambda_L}$ 
and define $f(t,\xi)$ as the solution to the  heat equation (\ref{eq:heateq})
with initial data
$f(0,\xi):=\tilde{T}(0,\xi)$.  Theorem \ref{th:maintprof} implies the following
bound for the error:
$|\Tobs(t,\xi)- \tilde{T}(t,\xi)|\le C\mathcal{E} L t_0^{-3/2}\sum_{y\in\Z}
|\varphi(\xi-y)|$.
Since the difference $\Tpred-f$ is a solution to the heat equation, we also
obtain
\begin{align}
  |\Tpred(t,\xi) -f(t,\xi)|\le \sup_{\xi'} |\Tpred(0,\xi') -f(0,\xi')| \le
C\mathcal{E} L t_0^{-3/2}
  \sup_{\xi'} \sum_{y\in\Z} |\varphi(\xi'-y)| \, .
\end{align}

Hence, it suffices to study the difference $f(t,\xi)-\tilde{T}(t,\xi)$.  For any
vector $\psi\in \C^{\Lambda_L}$, the Fourier transform of $h(\xi):=\sum_{y\in
\Z} \varphi(\xi-y) \psi_{y \bmod \Lambda_L}$ satisfies
  \begin{align}
  &  \FT{h}(n) = \frac{1}{L}\int_{L\T}\!\rmd \xi\, \rme^{-\ci 2 \pi n\cdot
\xi/L} 
    \sum_{x\in \Lambda_L} \psi_x \sum_{m\in \Z} \varphi(\xi-x+m L)
     \nonumber \\ & \quad
   =  \sum_{x\in \Lambda_L} \psi_x 
   \sum_{m\in \Z}   \int_{\T}\!\rmd q\, \rme^{-\ci 2 \pi n\cdot q} 
     \varphi(L(q+m-x/L))
     \nonumber \\ & \quad
   =  \sum_{x\in \Lambda_L} \psi_x 
 \int_{\R}\!\rmd q\, \rme^{-\ci 2 \pi n\cdot q} 
     \varphi(L(q-x/L))
     \nonumber \\ & \quad
   =  \frac{1}{L} \FT{\varphi}(n/L) \sum_{x\in \Lambda_L} \psi_x \rme^{-\ci 2
\pi n\cdot x/L}\, .
  \end{align}
Since $h\in C^{(2)}$, its Fourier transform is pointwise invertible, and thus at
every $\xi$ we then have 
  \begin{align}\label{eq:FTh}
   h(\xi)= \sum_{n\in \Z} \rme^{\ci 2 \pi n\cdot \xi/L}  \frac{1}{L}
\FT{\varphi}\!\left(\frac{n}{L}\right)
    \FT{\psi}\!\left(\frac{n \bmod \Lambda_L}{L}\right)\, .
  \end{align}
Therefore, using the definition of $\tau_0$ and (\ref{eq:diffsolFT}) to
represent $f$, we have
 \begin{align}
& f(t,\xi)-\tilde{T}(t,\xi) 
     \nonumber \\ & \quad
= \frac{1}{L}
\sum_{n\in \Z} \rme^{\ci 2 \pi n\cdot \xi/L}
\FT{\varphi}\!\left(\frac{n}{L}\right)
\left(\rme^{-t \kappa_L (2 \pi n/L)^2}-\rme^{-t \FT{D}(k)}\right) \rme^{-t_0
\FT{D}(k)} a(k) \sum_{y\in \Lambda_L} \rme^{-\ci 2 \pi k\cdot
y}\int_0^\infty\!\rmd s\, g_{s,y}\, ,   
 \end{align}
where $k$ is a shorthand for $(n \bmod \Lambda_L)/L$.  Here $\FT{D}(k) = 4 
\sum_{y\in \Lambda_L} \tilde{p}_y \sin^2(\pi k\cdot y)$ and thus
\begin{align}
  \kappa_L (2\pi k)^2-\FT{D}(k) = 4  \sum_{y\in \Lambda_L} \tilde{p}_y ((\pi
k\cdot y)^2-\sin^2(\pi k\cdot y))
  \le k^4 \frac{4 \pi^4}{3}\sum_{y\in \Lambda_L} \tilde{p}_y y^4\, ,
\end{align}
since $0\le x^2-\sin^2 x \le x^4/3 $ for any $x \in \R$.  Here $\sum_{y\in
\Lambda_L} \tilde{p}_y y^4$ is bounded by the $L$-independent constant $\gamma
C_0 \delta_0^{-1} \sum_{m=1}^\infty m^4 \rme^{-\gamma_2 m}<\infty$ and 
thus we can now conclude that there is a constant $c'>0$, independent of $L$ and
the initial state, such that 
$0\le \kappa_L (2\pi k)^2-\FT{D}(k) \le c' k^4$.  On the other hand, by
Proposition \ref{th:Greenprop}
for any $k\in \Lambda_L^*$ either $a(k)=0$ or $\FT{D}(k)\ge c'_3 k^2$ and
$0<a(k)\le c'_0$.  
Together with Proposition \ref{th:gtxprop} it follows that
 \begin{align}
& |f(t,\xi)-\tilde{T}(t,\xi)|\le \frac{c'_0}{\gamma} L \mathcal{E} \biggl[ c'
\norm{\varphi}_1  \int_{\Lambda_L^*} \!\rmd k\, 
t k^4 \rme^{-(t+t_0) c'_3 k^2} + \frac{1}{L}\sum_{|n| \ge L/2}
\left|\FT{\varphi}(n/L)\right|
\biggr]\, ,
\end{align}
where we have first treated separately the sum over $n\in\Lambda_L$ for which
$n/L=k$.  Here $\int_{\Lambda_L^*} \!\rmd k\, 
t k^4 \rme^{-(t+t_0) c'_3 k^2}$ can be estimated as in the proof of Theorem
\ref{th:maintprof}, which implies that it is bounded by a constant times
$t_0^{-3/2}$.  Therefore, collecting the above three estimates together, and
readjusting the constant $C$,
we find that (\ref{eq:Taccuracy}) holds.

To prove the final statement, assume additionally that $\varphi\ge 0$,
$\int_\R\!\rmd x\, \varphi(x)=1$, and $\FT{\varphi}(p)=0$ for all $|p|\ge
\frac{1}{2}$.  Then for any $\xi\in L\T$ we can apply (\ref{eq:FTh}) with
$\psi_x=1$ and conclude 
\begin{align}
 & \sum_{y\in \Z} |\varphi(\xi-y)|  = \sum_{y\in \Z} \varphi(\xi-y) = 
  \sum_{n\in \Z} \rme^{\ci 2 \pi n\cdot \xi/L}  \frac{1}{L}
\FT{\varphi}\!\left(\frac{n}{L}\right)
  \sum_{x\in \Lambda_L} \rme^{-\ci 2 \pi n\cdot x/L} 
     \nonumber \\ & \quad
  = \sum_{m\in \Z} \rme^{\ci 2 \pi m\cdot \xi}  \FT{\varphi}(m) =
\FT{\varphi}(0) = 1\, .
\end{align}
As now $\norm{\varphi}_1=1$ and  $\sum_{|n| \ge L/2}
|\FT{\varphi}(n/L)|=0$, the second bound  (\ref{eq:Taccuracy2}) follows.
\end{proofof}

\section{Proving complete thermalization via local dynamic replicas?}
\label{sec:LTE}

As a conclusion, let us present a local version of the dynamic replica method
introduced in Sec.~\ref{sec:global} and comment on how this might be used to
extend the results derived in Sec.~\ref{sec:Tprofile} into a proof of complete
thermalization.  For the present homogeneous system with periodic boundary
conditions, the benefits of using a local version of the replicas are perhaps not
immediately apparent.  However, for any system which is not totally
translation invariant, either due to boundary effects or to inhomogeneities, the
use of local replicas should be helpful since it allows using different local
dynamics for different lattice sites.  For instance, near the boundary it will
be necessary to incorporate the correct boundary conditions to properly account
for reflection and absorption at the boundary, whereas in the ``bulk'' one could
use the simpler periodic dynamics.   At
present, this is mere speculation, and it  remains to be seen how well such a 
division can be implemented in practice.

For simplicity, let us again assume that the harmonic interactions have a finite
range $\rphi$.  The lattice size is still assumed to be $L\gg 1$, but the local
dynamics will live on a smaller lattice $\Lambda_R$. For convenience, we assume
that $R$ is odd and satisfies $ \rphi\le R \le L/2$ which implies that $R+
\rphi \le L$.  The replicated generating functional is then defined as
\begin{align}
  h_t(\zeta;R) := \biggl\langle\exp\biggl[\ci \sum_{x\in \Lambda_L,y\in
\Lambda_R,i=1,2} \zeta^i_{x,y} X(t)^i_{[x+y]_L}\biggr]\biggr\rangle\, ,\quad
t\ge 0,\ \zeta\in \R^{2\times\Lambda_L\times \Lambda_R}\, ,
\end{align}
where we recall the notation $[x]_L$ used for projection onto $\Lambda_L$.  Let
us also drop ``$R$'' from the notation from now on.  If $y\in \Lambda_R$ and $z\in \Lambda_L$ are such that 
$[y-z]_L \ne y-z$, then $L/2 \le |y-z|\le L$ and thus $|[y-z]_L |\ge L-|y-z|\ge (L-R)/2 \ge \rphi/2$ implying $\Phi([y-z]_L )=0$.
Hence, 
$\sum_{z\in \Lambda_L}\Phi([y-z]_L ) X^1_{[x+z]_L} = \sum_{z\in \Lambda_L}\Phi(y-z) X^1_{[x+z]_L} 
=  \sum_{y'\in \Lambda_R} \sum_{\tau\in \set{-1,0,1}} 
\Phi(y-y'-\tau R) X^1_{[x+y'+\tau R]_L}$
for any $y\in \Lambda_R$.  Since for any $y,y'\in \Lambda_R$ also $\sum_{\tau\in \set{-1,0,1}} 
\Phi(y-y'-\tau R) = \Phi([y-y']_R )$, we then obtain as in
Sec.~\ref{sec:global} the following evolution equation
\begin{align}
  & \partial_t h_t(\zeta) = \frac{\gamma}{2} \sum_{x_0\in \Lambda_L} 
  \left( h_t(\sigma_{x_0} \zeta )-h_t(\zeta )-(\sigma_{x_0}-1) \zeta\cdot
\nabla_\zeta h_t  (\zeta )\right)
  - (\mathcal{M}_\gamma \zeta)\cdot \nabla_\zeta h_t  (\zeta ) 
    \nonumber \\ & \quad 
- \sum_{x\in \Lambda_L} \sum_{y',y\in \Lambda_R} \sum_{\tau=\pm 1}
\Phi(y'-y+\tau R) \zeta^2_{xy}
\mean{\ci (X^1_{[x+\tau R+y']}-X^1_{[x+y']}) \rme^{\ci \sum_{x_0,y_0,i}
\zeta^i_{x_0,y_0} X(t)^i_{x_0+y_0}}} \,,
\end{align}
where the last term collects the terms left over from the periodic harmonic evolution
in the replicated direction.  The matrix operations in the equation are defined
analogously to those appearing in Sec.~\ref{sec:global} as
\begin{align}
  (\sigma_{x_0} \zeta)_{xy}^i := \begin{cases}
                 -\zeta_{xy}^i\, , & \text{if }i=2 \text{ and }[x+y]_L=x_0\, ,
\\
                 \zeta_{xy}^i\, , & \text{otherwise}\, ,
               \end{cases}
\end{align}
and, similarly to (\ref{eq:defMM}),
\begin{align}
  \mathcal{M}_\gamma := \bigoplus_{x_0\in \Lambda_L} M_\gamma^{(x_0)}\, , \quad
  (M_\gamma^{(x_0)}\zeta)_{xy}^i  := \cf(x{=}x_0)\,
(M_\gamma\zeta_{x_0\cdot}^{\cdot})_{y}^i\, , \quad
  M_\gamma := \begin{pmatrix}
                0 & \Phi_R \\
                -1 & \gamma 1
              \end{pmatrix}\, .
\end{align}
The correction term can be expressed via derivatives of $h_t$ using the matrices
\begin{align}
  (D_R\zeta^2)_{xy} :=  \sum_{y'\in \Lambda_R} \sum_{\tau=\pm 1} \Phi(y-y'+\tau
R) (\zeta^2_{[x-\tau R],y'}-\zeta^2_{xy'})\, ,
 \qquad  \mathcal{D} := \begin{pmatrix}
                0 & D_R\\
                0&0
              \end{pmatrix}\, .
\end{align}
As in Sec.~\ref{sec:global}, we denote $\zeta_s := \rme^{-s  \mathcal{M}_\gamma}
\zeta$ and obtain the equality
\begin{align}
  & h_t(\zeta)=h_0(\zeta_t)-\int_0^t\! \rmd s\,(\mathcal{D} \zeta_s)\cdot
\nabla_\zeta h_{t-s}  (\zeta_s) 
    \nonumber \\ & \quad
  + 
  \int_0^t\! \rmd s\,
  \frac{\gamma}{2} \sum_{x_0\in \Lambda_L} 
  \left( h_{t-s}(\sigma_{x_0} \zeta_s)-h_{t-s}(\zeta_s)-(\sigma_{x_0}-1)
\zeta_s\cdot \nabla_{\zeta} h_{t-s}  (\zeta_s)\right)
 \,.
\end{align}

We recall the definition of the local statistics generating function and choose
$R_0=R$ here.  Explicitly, $f_{t,x_0}(\xi):= h_t(\zeta[\xi,x_0])$ with
$\zeta[\xi,x_0]_{xy}^i := \cf(x=x_0) \xi^i_y$ for
$\xi \in \R^{2\times \Lambda_{R}}$, $x_0\in \Lambda_L$.  We denote $\xi_s :=
\rme^{-s M_\gamma} \xi$, for which clearly $\zeta[\xi,x_0]_s =
\zeta[\xi_s,x_0]$.  Similarly, setting $(S_{y_0}\xi)^i_y := (-1)^{\cf(i=2,
y=y_0)} \xi^i_y$ implies that $(1-\sigma_{x'})\zeta[\xi,x_0]=\sum_{y_0\in
\Lambda_R}\cf(y_0=[x'{-}x_0]_L)$ $\times\zeta[(1{-}S_{y_0})\xi,x_0]$.  Then it
is straightforward to check that $f_{t,x_0}$ satisfies
\begin{align}
  & ((1-\mathcal{C})f)_t(\xi)= f_0(\xi_t)-\int_0^t\! \rmd s \left.(\mathcal{D}
\zeta_s)\cdot \nabla_\zeta h_{t-s}  (\zeta_s)\right|_{\zeta_s=\zeta[\xi_s,x_0]}
\, , \quad\text{with } 
    \nonumber \\ &
    (\mathcal{C}f)_t(\xi) :=
  \int_0^t\! \rmd s\,
  \frac{\gamma}{2} \sum_{y_0\in \Lambda_R} 
  \left( f_{t-s}(S_{y_0} \xi_s)-f_{t-s}(\xi_s)-(S_{y_0}-1) \xi_s\cdot
\nabla_{\xi} f_{t-s}  (\xi_s)\right)
 \,.
\end{align}
The operator $\mathcal{C}$ is essentially the same as the one appearing in
Sec.~\ref{sec:global}.   For the sake of argument, suppose that we could extend
the strong estimate in (\ref{eq:TdiffersfromE}) and show that, if $g$ is
exponentially decaying in time, then $(1-\mathcal{C})^{-1} g$ is always a sum of
an equilibrium generating function 
and a term which is exponentially small on the $R$-diffusive time scale, in $t
R^{-2}$.  Then the above formula would imply that also  $f$ has this property,
i.e., that strong local equilibrium holds.  The main additional hurdle for such
analysis is to find a separate control for the current term, $\int_0^t\! \rmd
s\,(\mathcal{D} \zeta_s)\cdot \nabla_\zeta h_{t-s}  (\zeta_s)$, since the term
needs to be sufficiently small and slowly varying for the local equilibrium
approximation to hold.  Proving this might be possible with an iterative
argument. However, 
further developments are needed to put any such scheme on solid ground. 

\section{On applications to other dynamical systems}
\label{sec:anharm}

Although the result in Theorem \ref{th:maintprof} requires sufficiently large
lattice sizes and times, it does not involve taking
any direct scaling limits.  This is somewhat surprising considering that up to
know nearly all mathematically rigorous work on energy diffusion in particle
systems of the present kind has relied on control of either a hydrodynamic
scaling limit, or as a middle step, a kinetic scaling limit leading to
a Boltzmann equation, see for instance \cite{erdyau99,erdyau05a,erdyau05b}.  

It is thus fair to ask if the present results are just particular properties of
the harmonic particle chain with velocity flips.  Indeed, it is unlikely
that the above computations for the velocity flip model can directly be
carried over to other models, such as purely Hamiltonian dynamical systems.
In the present case the evolution equation of second moments is closed, in the
sense that it does not depend on any higher order moments, and hence we do not
need to invert the full evolution equation of the characteristic function,
(\ref{eq:htiter}), but instead we can use an equation derived for its second
derivatives evaluated at zero, (\ref{eq:Tevoleq}).  In addition, for the
present system it
is straightforward to separate a part of the generator related to
the ``perturbation'' (the flips), and use this to produce a spectral gap for
the exponentiated linear term: this results in the change of
$\rme^{-t\mathcal{M}_0}$ to $\rme^{-t\mathcal{M}_\gamma}$ and is
responsible for the exponential decay in time of the memory kernel
$p_{t,x}$ in (\ref{eq:Tevoleq}).

To give an indication of the structure for other dynamical systems, we give below
an outline of an application of the dynamic replica method for an anharmonic
particle chain.  The application mimics the standard perturbation theory, and
should be taken with a grain of salt: there could well exist another way of
organizing the interactions so that the replica evolution semigroup
has better decay properties.  For this reason, we skip
all technical details of the computations below.

We consider the Hamiltonian system obtained by replacing the earlier velocity flips
by an anharmonic pinning potential $\gamma q_x^4$.  Explicitly, fix $L$
and let $\Phi_L$ and $H_L$ be defined as before, by
(\ref{eq:defphiL}) and (\ref{eq:defHLandGL}).
Define then for $X\in \R^{\Lambda_L}\times
\R^{\Lambda_L}$
\begin{align}
  & V_L(X)  := \gamma
  \sum_{x\in \Lambda_L} \frac{1}{4} (X_x^1)^4\, , 
\end{align}
choose some coupling $\gamma>0$, and set
\begin{align}
\anhH(X) := H_L(X) + \gamma V_L(X)
= \sum_{x\in \Lambda_L} \frac{1}{2} p_x^2
  + \sum_{x',x\in \Lambda_L} \frac{1}{2} (\Phi_L)_{x'x} q_{x'} q_x 
  + \gamma  \sum_{x\in \Lambda_L} \frac{1}{4} q_x^4 
\, .
\end{align}
The corresponding Hamiltonian evolution equations are 
\begin{align}
 \partial_t q(t)_x = p(t)_x\, , \qquad
 \partial_t p(t)_x = -(\Phi_L q(t))_x - \gamma q(t)_x^3\, .
\end{align}
For any initial data $q(0),p(0)$ there is a unique differentiable solution to
these equations, and this defines $X(t;X(0))$. 
We choose some random distribution 
$\mu_0$ for the initial data $X(0)$ and use this to define a random
vector $X(t)$.
The corresponding replicated generating function is
\begin{align}
h_t(\zeta) := \left\langle \rme^{\ci Y_t}\right\rangle \, ,\qquad
Y_t:=\sum_{x,y,i} \zeta^i_{x,y} X(t)^i_{x+y} \, ,
\end{align}
and it satisfies an evolution equation
\begin{align}
  & \partial_t h_t(\zeta) = 
  \sum_{x,y\in \Lambda_L} \zeta^1_{xy} \mean{\ci X^2_{x+y} \rme^{\ci
Y_t}}-\sum_{x,y,z\in \Lambda_L} \zeta^2_{xy} 
 (\Phi_L)_{x+y,x+z}\mean{\ci X^1_{x+z} \rme^{\ci Y_t}} 
   \nonumber \\ & \quad
      - \gamma \sum_{x,y\in \Lambda_L} \zeta^2_{xy} \mean{\ci (X^1_{x+y})^3
\rme^{\ci Y_t}}
\, .
\end{align}

We close the equation by using the previous computations for the first two terms
and set the anharmonic term to act at the origin of the replicated
direction.  This yields
\begin{align}
  & \partial_t h_t(\zeta) = - (\mathcal{M}_0 \zeta)\cdot \nabla h_t  (\zeta)    
 + \gamma \sum_{x_0\in \Lambda_L} \Bigl(\sum_{y\in \Lambda_L}
\zeta^2_{x_0-y,y}\Bigr) 
 \partial_{\zeta^1_{x_0,0}}^3 h_t(\zeta) \, .
\end{align}
Proceeding as in Section \ref{sec:global}, we obtain the following Duhamel
formula for $h_t$: with $\zeta_t := \rme^{-t \mathcal{M}_0} \zeta$ we have
\begin{align}\label{eq:hanhiter}
  & h_t(\zeta) = h_0(\zeta_t)+ \gamma \int_0^t\! \rmd s\,
  \sum_{x\in \Lambda_L} \Bigl(\sum_{y\in
\Lambda_L}(\zeta_s)^2_{x-y,y}\Bigr) 
 \left.\partial_{\zeta^1_{x0}}^3 h_{t-s}\right|_{\zeta_s}
 \,.
\end{align}

The formula (\ref{eq:hanhiter}) could then be used as a starting point for
further analysis.  For instance, if the goal is the study of kinetic scaling
limits, which involve only times $0\le t\le \tau \gamma^{-2}$, $\tau>0$ fixed and $\gamma\ll 1$,
it could be iterated once to obtain
\begin{align}
  & h_t(\zeta) = h_0(\zeta_t)
   \nonumber \\ & \quad
  + \gamma 
  \sum_{z\in \Lambda_L^4,i\in\set{1,2}^4} \sum_x
  \Bigl(\sum_y (\zeta_t)_{x-y,y+z_4}^{i_4}\Bigr)
 \partial_{\zeta^{i_1}_{x z_1}}\!
 \partial_{\zeta^{i_2}_{x z_2}}\!
 \partial_{\zeta^{i_3}_{x z_3}} h_0\bigr|_{\zeta_t}
  \int_0^t\! \rmd s\, \left(\rme^{s M_0}\right)^{2,i_4}_{0,z_4}
   \prod_{j=1}^3\left(\rme^{-s M_0}\right)^{i_j,1}_{z_j,0}
   \nonumber \\ & \quad
  + \gamma^2 \int_0^t\! \rmd t'\,
  \sum_{z\in \Lambda_L^4,i\in\set{1,2}^4} \sum_{x',x}
   \Bigl(\sum_{y'} (\zeta_{t'})_{x'-y',y'}^{2}\Bigr)
   \Bigl(\sum_y (\zeta_{t'})_{x-y,y+z_4}^{i_4}\Bigr)
\partial_{\zeta^{i_1}_{x z_1}}\!
 \partial_{\zeta^{i_2}_{x z_2}}\!
 \partial_{\zeta^{i_3}_{x z_3}} \partial^3_{\zeta^{1}_{x'0}}
  h_{t-t'}\bigr|_{\zeta_{t'}}
   \nonumber \\ & \qquad \times
  \int_0^{t'}\! \rmd s\, \left(\rme^{s M_0}\right)^{2,i_4}_{0,z_4}
   \prod_{j=1}^3\left(\rme^{-s M_0}\right)^{i_j,1}_{z_j,0}
   \nonumber \\ & \quad
  + 3 \gamma^2 \int_0^t\! \rmd t'\,
  \sum_{z\in \Lambda_L^4,i\in\set{1,2}^4} \sum_{x',x}
   \Bigl(\sum_y (\zeta_{t'})_{x-y,y+z_4}^{i_4}\Bigr)
   \partial_{\zeta^{i_1}_{x z_1}}\!
 \partial_{\zeta^{i_2}_{x z_2}}\!
 \partial^3_{\zeta^{1}_{x'0}}
  h_{t-t'}\bigr|_{\zeta_{t'}}
   \nonumber \\ & \qquad \times
 \cf(i_3=2, z_3=x'-x)
  \int_0^{t'}\! \rmd s\, \left(\rme^{s M_0}\right)^{2,i_4}_{0,z_4}
   \prod_{j=1}^3\left(\rme^{-s M_0}\right)^{i_j,1}_{z_j,0}
 \,.
\end{align}
The new terms all have additional decay arising from
the integration over $s$ which involves four oscillating factors.
This should facilitate rigorous analysis
of the kinetic scaling limits for a suitable class of initial states $\mu_0$.

However, let us stress once more that this approach might not be
optimal since it uses the unmodified free evolution in the replicated
directions.  Settling the question will likely require
careful consideration of which initial data to allow and of what function
space to use for the solutions $h_t$.

\appendix

\section{Local dynamics}\label{sec:applocal}

In this section we collect some basic properties of the semigroup $\rme^{-t
M_\gamma}$, for
\begin{align}
  M_\gamma := \begin{pmatrix}
                0 & \Phi_L \\
                -1 & \gamma 1
              \end{pmatrix} \, .
\end{align}
Since $\Phi_L$ is invariant under periodic translations, it can be diagonalized
by Fourier transform.
A discrete Fourier transform then yields, for $k\in \Lambda_L^*$, 
\begin{align}
     \FT{M}_\gamma(k) := \begin{pmatrix}
                0 & \omega(k)^2 \\
                -1 & \gamma
              \end{pmatrix}       \, .                                          
\end{align}
We make a Jordan decomposition of this matrix. The two eigenvalues, labeled by
$\sigma=\pm 1$, are
\begin{align}
\mu_\sigma(k) := \frac{\gamma}{2} + \sigma \sqrt{(\gamma/2)^2-\omega(k)^2}
  =\frac{\gamma}{2} + \ci \sigma \sqrt{\omega(k)^2-(\gamma/2)^2}\, ,
\end{align}
where the first formula is used if $\gamma>2 \omega(k)$, and the second formula
otherwise.  The case $\gamma=2 \omega(k)$ has a nontrivial Jordan block and
needs to be treated separately.  Using the decomposition, it is straightforward
to conclude that
\begin{align}
  (\rme^{-t M_\gamma})^{i'i}_{y'y} = \int_{\Lambda_L^*} \! \rmd k\, \rme^{\ci
2\pi k\cdot (y'-y)} (\rme^{-t \FT{M}_\gamma(k)})^{i'i}
\end{align}
where
\begin{align}
  & \rme^{-t \FT{M}_\gamma(k)} = \sum_{\sigma=\pm 1} \rmê^{-t \mu_\sigma(k)}
P_\sigma(k)\, ,\\
  & P_\sigma(k) :=  \left.\frac{1}{\mu_\sigma - \mu_{-\sigma}}
  \begin{pmatrix}
                -\mu_{-\sigma} & \omega^2 \\
                -1 & \mu_{\sigma}
              \end{pmatrix} \right|_{\mu_\pm := \mu_\pm(k), \omega:= \omega(k)},
\quad\text{if } |\omega(k)|\ne \frac{\gamma}{2}\, ,
  \end{align}
and taking the formal limit $\omega(k)\to \gamma/2$ yields the correct
expressions also for the degenerate case.
The eigenvalues satisfy, for $\omega(k)\ge \gamma/2$  and dropping the variable
$k$ from the notations,
\begin{align}
  & \mu_\sigma =\frac{\gamma}{2} + \ci \sigma \sqrt{\omega^2-(\gamma/2)^2}\,
,\quad
  |\mu_\sigma|=\omega\, ,\quad \mu_\sigma^* =  \mu_{-\sigma}\, ,\quad
\mu_\sigma-\mu_{-\sigma}= 2 \ci \sigma\sqrt{\omega^2-(\gamma/2)^2}\, .
\end{align}
If $\omega(k)< \gamma/2$, then
\begin{align}
  & \mu_-<\omega< \mu_+\, , \quad \mu_+ \in \gamma [\frac{1}{2},1]\, ,\quad
  \mu_- = \frac{\omega^2}{\mu_+}\in \frac{\omega^2}{\gamma} [1,2]
  \, ,\quad \mu_\sigma-\mu_{-\sigma}= 2 \sigma\sqrt{(\gamma/2)^2-\omega^2}\, .
  \end{align}

\section{Solution of the lattice renewal equation}
\label{sec:diffkernel}

In this appendix, we recall some basic mathematical properties of the Green's
function solution of the renewal equations encountered in the main text.  By Corollary \ref{th:ptxcoroll} the
corresponding functions $p_{t,x}$ satisfy all of the assumptions below. 
\begin{proposition}\label{th:renewalprop} 
Consider some $L\ge 1$ and $p_{t,x}$, given for $t\ge 0$, $x\in \Lambda_L$. 
Suppose that there are constants $C_0,\delta_0,\gamma_2>0$ such that all of the
following statements hold:
  \begin{enumerate}
  \setlength{\itemsep}{0 pt}
    \item $t\mapsto p_{t,x}$ belongs to $C^{(1)}([0,\infty))$ for all $x\in
\Lambda_L$.
    \item $p_{t,x}\ge 0$ and $|\partial_t p_{t,x}|, p_{t,x}\le C_0
\rme^{-\gamma_2 |x|-\delta_0 t}$ for all $x\in \Lambda_L$ and $t> 0$.
    \item $\int_0^\infty \!\rmd t \sum_{x\in \Lambda_L} p_{t,x} =1$ and
$\int_0^{t_0} \!\rmd t \sum_{x\in \Lambda_L} p_{t,x} <1$ for all $t_0\ge 0$.
  \end{enumerate}
  Then there is a unique continuous function $G\in C([0,\infty)\times
\Lambda_L)$ for which
  \begin{align}\label{eq:Arenewal}
    G(t,x) = p_{t,x} + \int_0^t \!\rmd s\, \sum_{y\in\Lambda_L} G(t-s,x-y)
p_{s,y} \, ,\quad t\ge 0,\  x\in \Lambda_L\, .
  \end{align}
  In addition, this $G$ is non-negative, bounded, continuously differentiable,
and has the following pointwise integral representation, valid for any $\vep>0$,
  \begin{align}\label{eq:defGreen}
    G(t,x) = p_{t,x} + \int_{\Lambda_L^*} \!\rmd k\, \rme^{\ci 2\pi k\cdot x}
\int_{\vep-\ci \infty}^{\vep+\ci \infty}\! \frac{\rmd \lambda}{2\pi \ci} \,
\rme^{\lambda t} \frac{\FT{p}(\lambda,k)^2}{1-\FT{p}(\lambda,k)} \, .
  \end{align}
  In the integrand, the function $\FT{p}$ denotes the Laplace-Fourier transform
of $p$, defined for all $k\in\Lambda_L^*$ and $\lambda\in \C$ with 
  $\re \lambda>-\delta_0$ by 
  \begin{align}\label{eq:defFTp}
    \FT{p}(\lambda,k) := \int_0^\infty\! \rmd s \, \sum_{y\in \Lambda_L} p_{s,y}
\rme^{-s\lambda}\rme^{-\ci 2 \pi k \cdot y}\, .
  \end{align}
  It satisfies the following properties:
  \begin{enumerate}
  \setlength{\itemsep}{0 pt}
    \item\label{it:FTp1} The map $\lambda\mapsto \FT{p}(\lambda,k)$ is an
analytic function on the half plane $\re \lambda>-\delta_0$ for any $k$, and 
there  all of its $\lambda$-derivatives can be computed by differentiating the
integrand.
   \item\label{it:FTp3} For any $\vep>0$, there is $c_\vep>0$ such that
$|\FT{p}(\vep+\ci\alpha ,k)|\le 1-c_\vep$ for all $\alpha\in \R$ and
$k\in\Lambda_L^*$.
    \item\label{it:FTp2} There is $C'>0$, which depends only on the input
constants $C_0,\delta_0,\gamma_2>0$, such that
    $|\FT{p}(\lambda,k)|\le C' (1+|\lambda|)^{-1}$ whenever $\re \lambda\ge
-\delta_0/2$.
   \end{enumerate} 
  In particular, the integral in (\ref{eq:defGreen}) is always absolutely
convergent.
\end{proposition}

\begin{proof}
  Assume that the constants $C_0,\delta_0,\gamma_2>0$, which will be called the
\defem{input parameters}, and $p_{t,x}$ are given as above.  Then the integrand
in the definition (\ref{eq:defFTp}) is bounded by $C_0 \rme^{-\gamma_2
|y|-s(\delta_0+\re \lambda)}$, hence it is bounded by an $L^1$-function for any
compact subset of the half-plane $\re\lambda>-\delta_0$.  Since the integrand is
an entire function of $\lambda$, we can now conclude (for instance via
Morera's theorem) that for any $k$ formula
  (\ref{eq:defFTp}) defines an analytic function  $\FT{p}(\lambda,k)$ on the
half plane.  In addition, we can compute the derivatives by differentiation
inside the integral for these values of $\lambda$; this can be concluded for
instance by relying on Cauchy's integral formula for derivatives of an analytic
function, and then using Fubini's theorem.  In particular, $\FT{p}$ satisfies
the properties in item \ref{it:FTp1}.
  
If $\re \lambda \ge \vep>0$, then 
  $|\FT{p}(\lambda,k)|\le \int_0^\infty\! \rmd s \, \sum_{y\in \Lambda_L}
p_{s,y} \rme^{-s \vep}=:1-c_\vep$.  
  Applying the assumptions, then $c_\vep = \int_0^\infty\! \rmd s \, \sum_{y}
p_{s,y} (1-\rme^{-s\vep})
  =\vep\int_0^\infty\! \rmd s' \, \rme^{-s'\vep} \int_{s'}^\infty\! \rmd s \,
\sum_{y} p_{s,y}\ge 0$.  If $c_\vep=0$, then 
  $\int_{s'}^\infty\! \rmd s \, \sum_{y} p_{s,y}=0$ for almost every $s'>0$,
which contradicts the assumed normalization condition.
  Therefore, $\FT{p}$ satisfies also item \ref{it:FTp3}. In particular,
$1-\FT{p}(\lambda,k)$ has no zeroes on the positive right half-plane.
  
The assumed differentiability of $p$ allows to use partial integration to
conclude that for any $\lambda\ne 0$ with $\re\lambda>-\delta_0$ we have
\begin{align}
  \FT{p}(\lambda,k)
  = \frac{1}{\lambda} \sum_{y\in \Lambda_L} \rme^{-\ci 2 \pi k \cdot y}\left[
p_{0,y} 
 + \int_0^\infty\! \rmd s \,  \rme^{-s\lambda} \partial_s p_{s,y} \right] \, ,
\end{align}
which is bounded by $|\lambda|^{-1} C_0 
(1+1/(\delta_0+\re \lambda)) \sum_{y\in \Lambda_L} \rme^{-\gamma_2 |y|}$ where $
 \sum_{y} \rme^{-\gamma_2 |y|} \le 2/(1-\rme^{-\gamma_2})$.  If $|\lambda|\le
\delta_0/2$, we similarly obtain from the definition a bound
$|\FT{p}(\lambda,k)|\le 4\delta_0^{-1} C_0 /(1-\rme^{-\gamma_2})$.
Therefore, item \ref{it:FTp2} holds with $C':= 2(1+2\delta_0^{-1})^2 C_0
/(1-\rme^{-\gamma_2})$.

The above estimates imply a bound $\rme^{\vep t} (C')^2 c_\vep^{-1} (1{+}|\lambda|)^{-2}$
for the integrand in (\ref{eq:defGreen}). Thus it is
absolutely integrable, for any $\vep>0$.  In addition, the integrand is analytic
on the whole right half plane, and hence Cauchy's theorem allows to conclude
that its value does not depend on the choice of $\vep$.  We fix some value
$\vep>0$ and then define $G(t,x)$ by (\ref{eq:defGreen}) for all $t\ge 0$,
$x\in\Lambda_L$.  Relying on dominated convergence, we conclude that then $G$ is
continuous.  Also, the above bounds imply that the value of the integral in
(\ref{eq:defGreen}) is zero for any $t\le 0$ since in that case Cauchy's theorem
allows taking $\vep\to \infty$.  In particular, then $G(0,x)=p_{0,x}$ and thus
(\ref{eq:Arenewal}) holds at $t=0$.

In order to check that $G$ satisfies (\ref{eq:Arenewal}) also for $t>0$, we rely
on the following computation, whose steps can be justified by using Fubini's
theorem and the above observation about the vanishing of the integral for
negative $t$:
  \begin{align}
 &   \int_0^t\! \rmd s \, \sum_y p_{s,y} G(t-s,x-y) - \int_0^t\! \rmd s \,
\sum_y p_{s,y} p_{t-s,x-y}
 \nonumber \\ & \quad
    =  \int_0^t\! \rmd s \, \sum_y p_{s,y}
    \int_{\Lambda_L^*} \!\rmd k\, \rme^{\ci 2\pi k\cdot (x-y)}
    \int_{\vep -\ci \infty}^{\vep +\ci \infty}\! \frac{\rmd \lambda}{2\pi \ci}
\, \rme^{\lambda (t-s)} \frac{\FT{p}(\lambda,k)^2}{1-\FT{p}(\lambda,k)}
 \nonumber \\ & \quad
    =  \int_0^\infty\! \rmd s \, \sum_y p_{s,y}
    \int_{\Lambda_L^*} \!\rmd k\, \rme^{\ci 2\pi k\cdot (x-y)}
    \int_{\vep -\ci \infty}^{\vep +\ci \infty}\! \frac{\rmd \lambda}{2\pi \ci}
\, \rme^{\lambda (t-s)} \frac{\FT{p}(\lambda,k)^2}{1-\FT{p}(\lambda,k)}
 \nonumber \\ & \quad
    =  
    \int_{\Lambda_L^*} \!\rmd k\, \rme^{\ci 2\pi k\cdot x}
    \int_{\vep -\ci \infty}^{\vep +\ci \infty}\! \frac{\rmd \lambda}{2\pi \ci}
\, \rme^{\lambda t} \frac{\FT{p}(\lambda,k)^3}{1-\FT{p}(\lambda,k)}
  \nonumber \\ & \quad
    =  G(t,x)-p_{t,x} -    \int_{\Lambda_L^*} \!\rmd k\, \rme^{\ci 2\pi k\cdot
x}
    \int_{\vep -\ci \infty}^{\vep +\ci \infty}\! \frac{\rmd \lambda}{2\pi \ci}
\, \rme^{\lambda t} \FT{p}(\lambda,k)^2
\, .
  \end{align}
Since the final integral is equal to $\int_0^t\! \rmd s \, \sum_y p_{s,y}
p_{t-s,x-y}$, we have proven that (\ref{eq:Arenewal}) holds for all $t\ge 0$ and
$x\in \Lambda_L$.  

Hence, the above function $G$ provides a continuous solution to
(\ref{eq:Arenewal}).  Let us next prove that this solution is unique.  Consider
an arbitrary $t_0>0$ and the Banach space $X=C([0,t_0]\times \Lambda_L,\C)$
endowed with the sup-norm.  For $f\in X$ and $0\le t\le t_0$,
$x\in\Lambda_L$, set $(B f)(t,x) := \int_0^t \!\rmd s\,
\sum_{y\in\Lambda_L} f({t-s},{x-y}) p_{s,y}$.  
Then $B f$ is continuous and bounded by $r_0 \norm{f}$, where
$r_0 := \int_0^{t_0}\!\rmd s\, \sum_{y\in\Lambda_L} p_{s,y}$.  Hence
$\norm{B}\le r_0$ and it follows from the assumptions that $r_0<1$. Therefore,
the inverse of $1-B$ exists and is a bounded operator on $X$ defined by the
convergent Dyson series formula, $(1-B)^{-1} = \sum_{n=0}^\infty B^n$.  Suppose
$f$ is a continuous function which solves (\ref{eq:Arenewal}) and let $g$ denote
the restriction of $f$ to $[0,t_0]$. Then we have $g\in X$ and $(1-B)g$ is equal
to $h$, the restriction of $p$ to $[0,t_0]$.
  Thus necessarily $g=(1-B)^{-1}h$.  Hence for all $0\le t\le t_0$ we have
$f(t,x)=\sum_{n=0}^\infty (B^n h)(t,x)$, and as $t_0$ can be taken arbitrarily
large, this proves the uniqueness of the solution.  Since $B$ is obviously
positivity preserving and $h$ is nonnegative, this also implies that the unique
solution, coinciding with $G$ defined in (\ref{eq:defGreen}), is pointwise
nonnegative.
  As $p$ is assumed to be continuously differentiable, the continuous solution
$G$ to (\ref{eq:Arenewal}) is that, as well.
This concludes the proof of the Proposition.
\end{proof}

\begin{corollary}\label{th:renewcorr}
  Suppose that $p$ and $G$ are given as in Proposition \ref{th:renewalprop}. 
Then for any $h\in C([0,\infty)\times \Lambda_L,\C)$ the formula
  \begin{align}\label{eq:defgensol}
    f(t,x) := h(t,x) + \int_0^t\!\rmd s\, \sum_{y\in\Lambda_L} G(t-s,x-y)
h(s,y)\, , \quad t\ge 0,\ x\in \Lambda_L \, ,
  \end{align}
 defines the unique continuous solution to the equation
  \begin{align}\label{eq:genrenew}
    f(t,x) = h(t,x) + \int_0^t\!\rmd s\, \sum_{y\in\Lambda_L} p_{s,y}
f(t-s,x-y)\, .
  \end{align}
\end{corollary}
\begin{proof}
  Since $h$ is continuous, dominated convergence theorem immediately implies
that (\ref{eq:defgensol}) defines a continuous function. Then a straightforward
application of Fubini's theorem and (\ref{eq:Arenewal}) proves that
(\ref{eq:genrenew}) holds for all $t,x$. The uniqueness can be proven via the
same argument which was used in the proof of Proposition \ref{th:renewalprop}.
\end{proof}

\subsection*{Acknowledgments}

The research of J.~Lukkarinen was partially supported by the Academy of
Finland and partially by the European Research Council (ERC) Advanced
Investigator Grant 227772.  I thank Fran\c{c}ois Huveneers,
Wojciech de Roeck, and Herbert Spohn for discussions and suggestions and Janne Junnila for an
introduction to mathematical ergodic theory.  I am also grateful to the
anonymous reviewer for pointing out additional references and improvements of
representation.

\end{document}